\documentclass[journal]{IEEEtran}
\usepackage{fancyhdr,graphicx,amsmath,amssymb, amsfonts}
\usepackage{dsfont}
\usepackage{array}
\usepackage[caption=false,font=normalsize,labelfont=rm,textfont=rm]{subfig}
\usepackage{textcomp}
\usepackage{stfloats}
\usepackage{url}
\usepackage{verbatim}
\usepackage{graphicx}
\usepackage[nocompress]{cite}
\usepackage{xcolor}
\usepackage{amsthm}
\usepackage[ruled,vlined,linesnumbered]{algorithm2e}
\usepackage{nomencl}
\usepackage{algpseudocode}
\DeclareMathOperator*{\argmin}{arg\,min}
\usepackage{makecell}
\usepackage{multicol, multirow}
\usepackage{kotex}

\newtheorem{lemma}{Lemma}

\newtheorem{remark}{Remark}
\usepackage[colorlinks=true, linkcolor=blue, citecolor=blue, urlcolor=blue]{hyperref}
\allowdisplaybreaks

\ifCLASSINFOpdf
\else
\fi
\begin{document}
%
% paper title
% Titles are generally capitalized except for words such as a, an, and, as,
% at, but, by, for, in, nor, of, on, or, the, to and up, which are usually
% not capitalized unless they are the first or last word of the title.
% Linebreaks \\ can be used within to get better formatting as desired.
% Do not put math or special symbols in the title.
\title{Low-Complexity Beamspace Channel Denoiser for mmWave Massive MIMO with Low-Resolution ADCs}
%
%
% author names and IEEE memberships
% note positions of commas and nonbreaking spaces ( ~ ) LaTeX will not break
% a structure at a ~ so this keeps an author's name from being broken across
% two lines.
% use \thanks{} to gain access to the first footnote area
% a separate \thanks must be used for each paragraph as LaTeX2e's \thanks
% was not built to handle multiple paragraphs
%

\author{Hanyoung Park,~\IEEEmembership{Graduate Student Member,~IEEE,}
        Eunho Kim, 
        and~Ji-Woong Choi,~\IEEEmembership{Senior Member,~IEEE}% <-this % stops a space
\thanks{This work was supported by Institute of Information \& communications Technology Planning \& Evaluation (IITP) grant funded by the Korea government (MSIT) (No. RS-2024-00442085, No. RS-2024-00398157). 
\textit{(Corresponding author: Ji-Woong Choi.)}}
\thanks{H. Park is with the Department
of Electrical Engineering and Computer Science, Daegu Gyeongbuk Institute of Science and Technology (DGIST), Daegu 42988,
 South Korea (e-mail: prkhnyng@dgist.ac.kr).}% <-this % stops a space
 \thanks{E. Kim and J.-W. Choi are with the Department
of Electrical Engineering and Computer Science, Daegu Gyeongbuk Institute of Science and Technology (DGIST), Daegu 42988,
 South Korea, and also with Brain Engineering Convergence Research Center (BCC), DGIST, Daegu 42988, South Korea (e-mail: hunho0111@dgist.ac.kr; jwchoi@dgist.ac.kr).}% <-this % stops a space
% \thanks{Manuscript received April 19, 2005; revised August 26, 2015.}
}

\maketitle

\begin{abstract}
In this paper, we propose a low-complexity beamspace channel denoising algorithm for millimeter-wave (mmWave) massive multi-input multi-output (MIMO) systems with low-resolution analog-to-digital converters (ADCs). 
The proposed method exploits the inherent sparsity of mmWave channels in the beamspace domain and formulates the denoising problem as a Bayesian binary hypothesis testing under a Bernoulli-complex Gaussian prior. 
To capture the distortion induced by low-resolution ADCs in a complexity-efficient manner, thermal noise and quantization noise are jointly modeled as a composite noise.
Based on this modeling, a closed-form threshold value and a hard-thresholding-based denoising rule are derived to distinguish signal-dominant and noise-dominant components. 
The resulting algorithm avoids computationally intensive operations such as matrix inversion, iterative optimization, and parameter searching, and achieves near-linear computational complexity with respect to the number of antennas. 
Furthermore, a hardware-efficient very large-scale integration (VLSI) architecture is developed to enable practical deployment of the proposed algorithm, and is implemented on an AMD-Xilinx Kintex UltraScale+ KCU116 FPGA platform. The design incorporates hardware-aware simplifications and an efficient processing structure, leading to significantly lower latency and reduced hardware resource utilization compared to existing hardware implementations, along with sublinear scaling as the number of antennas increases. Extensive simulation results demonstrate that the proposed method achieves performance comparable to computationally intensive existing approaches while significantly reducing computational complexity.
\end{abstract}

% Note that keywords are not normally used for peerreview papers.

% For peer review papers, you can put extra information on the cover
% page as needed:
% \ifCLASSOPTIONpeerreview
% \begin{center} \bfseries EDICS Category: 3-BBND \end{center}
% \fi
%
% For peerreview papers, this IEEEtran command inserts a page break and
% creates the second title. It will be ignored for other modes.
\IEEEpeerreviewmaketitle

\section{Introduction}
% The very first letter is a 2 line initial drop letter followed
% by the rest of the first word in caps.
% 
% form to use if the first word consists of a single letter:
% \IEEEPARstart{A}{demo} file is ....
% 
% form to use if you need the single drop letter followed by
% normal text (unknown if ever used by the IEEE):
% \IEEEPARstart{A}{}demo file is ....
% 
% Some journals put the first two words in caps:
% \IEEEPARstart{T}{his demo} file is ....
% 
% Here we have the typical use of a "T" for an initial drop letter
% and "HIS" in caps to complete the first word.

With the development of wireless communication techniques, wireless connections are considered for various advanced services and applications such as virtual/augmented reality (VR/AR), smartglasses, autonomous driving, remote surgery, and so on~\cite{ref:ericsson, ref:6g}. This connectivity requires the transmission of high-volume data with low latency. To this end, abundant bandwidth is mandatory, and millimeter-wave (mmWave) communication systems are regarded as future solutions for next-generation communications~\cite{ref:ericssonmobreport2025q4}. However, mmWave has its own characteristics compared to other wireless communications, such as severe path loss, susceptibility to blockages, and limited diffraction~\cite{ref:mmwavetextbook}. These characteristics lead to a limited number of multipaths, resulting in spatial sparsity and highly directional channel responses~\cite{ref:mmwavework, ref:mmwaveblindsnr}. To overcome these challenges in mmWave propagation environments, massive multiple-input multiple-output (MIMO) architectures are employed, allowing for the compensation of these artifacts through directional beamforming~\cite{ref:mmwavemimo}.

However, adopting massive MIMO raises two challenges for mmWave systems. First, the computational complexity of signal processing is increased with the number of antennas. Since beamforming-based systems heavily depend on channel state information (CSI), and this estimate is utilized in beam selection~\cite{ref:ex1_beamselect}, precoding~\cite{ref:ex2_precoding}, equalization~\cite{ref:ex3_equalization}, scheduling~\cite{ref:ex4_scheduling}, and so on. To utilize the CSI for these applications, the channel estimation must reflect the changes in the dynamic wireless environments. However, as mentioned above, the channel changes at a fast rate since mmWave is sensitive to blockages or nearby obstacles due to its lack of multipaths~\cite{ref:pathloss}. Additionally, its channel coherence time is shorter because of its high frequency~\cite{ref:mmwavesp}. Accordingly, channel estimation must be performed with low latency to ensure the validity of the channel estimates. However, since the computational complexity of channel estimation increases with the number of antennas, the development of a low-latency channel estimation algorithm under overhead constraints remains an important challenge in mmWave MIMO systems.
In addition, the number of analog-to-digital converters (ADCs) in massive MIMO increases in proportion to the number of antennas. This can be a burden from the perspective of power consumption and costs~\cite{ref:mmwavesp}. Thus, the adoption of low-resolution ADCs is considered to overcome this challenge~\cite{ref:lowresadc, ref:lowresadc2}. However, the low-resolution of ADCs means a coarse quantization, so it is required to suppress the influence of quantization noise on the channel estimation.

\subsection{Prior Works}
Low-resolution ADCs offer significant advantages in terms of power consumption, chip area, and hardware cost. 
However, the severe quantization noise introduced by coarse quantization leads to substantial signal distortion, which has motivated extensive research efforts to mitigate its impact. One line of work explicitly models the nonlinear quantization effect to address this issue~\cite{ref:nonlinear1,ref:nonlinear2,ref:nonlinear3,ref:nonlinear4,ref:nonlinear5,ref:nonlinear6,ref:nonlinear7,ref:glqvbce,ref:nonlinear8}. 
The work in \cite{ref:nonlinear1} proposes a maximum likelihood (ML)-based channel estimation and detection scheme for 1-bit ADC MIMO systems. In \cite{ref:nonlinear2}, a likelihood-based signal recovery method applicable not only to 1-bit but also to general low-resolution ADCs is presented. Liu \textit{et al.} further develop an ML-based channel estimation framework that accounts for nonlinear quantization effects~\cite{ref:nonlinear3}. 
Similarly, the work in \cite{ref:nonlinear4} proposes a channel estimation method incorporating nonlinear quantization, while the work in \cite{ref:nonlinear5} combines nonlinear quantization modeling with soft symbol decoding for improved performance. 
In addition, the work in \cite{ref:nonlinear6} introduces a quantization-aware joint channel estimation and decoding scheme applicable to frequency-selective channels. 
Compressive sensing (CS)-based approaches that incorporate nonlinear distortion are also investigated in \cite{ref:nonlinear7, ref:glqvbce}, and their extensions to wideband systems are presented in \cite{ref:nonlinear8}. Despite their performance gains, these nonlinear modeling approaches often lead to mathematically intractable formulations and incur substantial computational complexity, which becomes particularly prohibitive in massive MIMO systems with a large number of antennas.

Since such nonlinear modeling can be somewhat complex in massive MIMO systems, various approaches based on the additive quantization noise model (AQNM)~\cite{ref:aqnm}, derived from Bussgang's decomposition\cite{ref:bussgang}, have been proposed. This model approximates the nonlinear, signal-dependent quantization noise as an additive form and has been widely adopted as an underlying assumption in many low-resolution ADC-based system studies~\cite{ref:aqnm_app1, ref:aqnm_app2, ref:aqnm_app3, ref:aqnm_app4}. 
The works in \cite{ref:aqnm0, ref:aqnm_extend, ref:aqnm1} proposed minimum mean-squared error (MMSE)-based channel estimation methods using Bussgang's decomposition. 
The work in \cite{ref:aqnm2} presented an mmWave channel estimation method that combines a CS-based approach with AQNM. 
Rao \textit{et al.} proposed a channel estimation method that mitigates quantization noise in sigma-delta ADC systems~\cite{ref:aqnm3}. 
The work in \cite{ref:sand} proposed a channel denoising algorithm that combines Stein's unbiased risk estimate (SURE) with Bussgang's decomposition. Srinivas \textit{et al.} proposed a channel estimation method that applies sparse Bayesian learning to Bussgang's decomposition~\cite{ref:aqnm4}. 
Ito \textit{et al.} proposed a joint channel estimation and data detection method based on expectation propagation (EP), where nonlinear distortion is approximated as Gaussian messages~\cite{ref:linear1}. Fesl \textit{et al.} proposed a linear MMSE estimator for 1-bit quantized systems using a Gaussian mixture model (GMM)~\cite{ref:aqnm5}.

Meanwhile, recent advances in deep learning have also been widely applied to the low-resolution ADC channel estimation problem~\cite{ref:sl1,ref:sl2,ref:sl3,ref:sl4,ref:sl5,ref:sl6,ref:sl7,ref:sl8,ref:sl9,ref:sl10,ref:sl11,ref:sl12,ref:ssl1,ref:ul1,ref:ldgec}. 
Gao \textit{et al.} proposed a deep neural network (DNN)-based channel estimation method for mixed-resolution ADC systems~\cite{ref:sl1}. The work in \cite{ref:sl2} interpreted 1-bit channel estimation as a deep learning-based mapping problem. The work in \cite{ref:sl3} proposed a channel estimation method that learns the mapping between signals and channels using a generative adversarial network (GAN). The work in \cite{ref:sl4} presented a convolutional neural network (CNN)-based channel estimation method for 1-bit quantization systems. Balevi \textit{et al.} proposed a GAN-based unsupervised method that utilizes approximate message passing (AMP) to enable channel estimation in high-dimensional settings with 1-bit ADCs~\cite{ref:sl5}. The work in \cite{ref:sl6} proposed a model-driven approach based on AMP. The work in \cite{ref:sl7} proposed a method that integrates a generative latent model with Bussgang linearization. The work in \cite{ref:sl8} proposed a channel estimation algorithm for 1-bit massive MIMO systems using a conditional GAN (CGAN), which was further refined in \cite{ref:sl9, ref:sl10} for mixed-resolution ADC systems and adaptive quantization schemes. The work in \cite{ref:sl11} proposed a 1-bit MIMO estimation algorithm that leverages the advantage of GAN-based inverse quantization mapping while mitigating training instability. Helmy \textit{et al.} proposed a channel estimation method that performs denoising by blindly estimating quantization noise using long short-term memory (LSTM) and attention mechanisms~\cite{ref:sl12}. Liu \textit{et al.} proposed a method that combines AQNM with self-supervised learning~\cite{ref:ssl1}. The work in \cite{ref:ul1} proposed a channel estimation method based on SURE and diffusion models. He \textit{et al.} formulated the wideband mmWave channel estimation problem as a two-dimensional CS problem and proposed an unsupervised channel denoising method using SURE and CNN~\cite{ref:ldgec}.
Although these approaches achieve strong performance, they often involve high algorithmic complexity or computationally burdensome structures. This issue becomes more pronounced in scenarios with a large number of antennas, and may limit real-time applicability, particularly in rapidly varying mmWave channels. Moreover, these methods are primarily designed as software-based signal processing algorithms without considering hardware implementation aspects.

Hardware implementation of MIMO channel estimation methods has also been extensively studied~\cite{ref:hw1,ref:hw2,ref:hw3,ref:hw4,ref:beaches,ref:liu, ref:hw5, ref:hw6, ref:hw7}. The works in \cite{ref:hw1, ref:hw2} proposed joint channel estimation and data detection schemes and implemented the corresponding hardware architectures. Damjancevic \textit{et al.} proposed a hardware-aware architecture that efficiently performs LS-based channel estimation~\cite{ref:hw3}. Chundi \textit{et al.} presented a channel estimation algorithm based on model-driven deep learning along with its hardware implementation~\cite{ref:hw4}. The work in \cite{ref:beaches} implemented hardware for mmWave beamspace channel estimation. 
Liu \textit{et al.} proposed a complexity-efficient channel estimator for sparse MIMO channels~\cite{ref:liu}.
Fu \textit{et al.} proposed an architecture capable of performing various MIMO signal processing operations, including channel estimation, with high energy efficiency~\cite{ref:hw5}. The work in \cite{ref:hw6} proposed an integrated design for synchronization and channel estimation in sub-THz band MIMO systems. The work in \cite{ref:hw7} designed a processor for efficient deep learning-based channel estimation. However, these approaches do not consider low-resolution ADCs for massive MIMO systems.

\subsection{Contributions}
Fast and computationally efficient channel estimation algorithms are essential, since mmWave channels exhibit rapid fluctuations~\cite{ref:mmwavetextbook}. However, existing approaches primarily focus on accuracy rather than real-time implementation, and often rely on iterative procedures or deep learning-based inference, which require significant computational resources. 
In addition, channel estimation methods for low-resolution ADC systems have rarely considered hardware implementation. Therefore, in this paper, we propose a low-complexity channel denoising algorithm for mmWave massive MIMO systems with low-resolution ADCs. The main contributions of this paper are summarized as follows:
\begin{itemize}
    \item We propose a novel channel denoising algorithm that exploits the sparsity of mmWave channels in the beamspace representation and statistical inference. The beamspace channel denoising problem is modeled as an element-wise binary hypothesis testing problem, where the noise is modeled as a composite noise consisting of thermal noise and quantization noise. Each beamspace element is modeled as either containing a signal or containing only noise, and if it is classified as noise-only, it is removed via a hard-thresholding-based denoising operation.
    \item The proposed method does not rely on computationally intensive operations such as matrix inversion or parameter searching, and achieves near-linear computational complexity in the number of antennas. Despite reducing computational complexity, it achieves performance comparable to existing computationally intensive methods.
    \item We propose a very large-scale integration (VLSI) architecture for the algorithm and present corresponding field-programmable gate array (FPGA) implementation results. The proposed method is refined with hardware-aware simplifications and a hardware-efficient structure, and the implementation results demonstrate lower resource utilization and shorter latency compared to conventional channel estimation methods.
\end{itemize}

\subsection{Paper Outline}
The rest of paper is organized as follows. Section~\ref{sec:systemmodel} introduces the system model. Section~\ref{sec:proposed} describes the low-complexity beamspace channel denoiser algorithm. Section~\ref{sec:simresults} shows the simulation results and performance evaluation.
Section~\ref{sec:hardware} describes the proposed VLSI design, its implementation, and corresponding analyses. Section~\ref{sec:conclusion} concludes this paper. 

\subsubsection*{Notation}
Matrices and vectors are denoted by boldface uppercase and lowercase letters, respectively.
For matrix $\mathbf{A}$, its $i$-th column vector is $\mathbf{a}_i$, and its $j$-th element is $A_{ij}$.
For vector $\mathrm{a}$, its $i$-th element is $a_i$.
$(\cdot)^T$ and $(\cdot)^H$ denote transpose and conjugate transpose, respectively.
$\mathbf{I}_N$ and $\mathbf{0}^{N\times M}$ mean $N\times N$ identity matrix and $N\times M$ zero matrix, respectively.
$\|\cdot\|$ represents the $\ell_2$ norm, and $\|\mathbf{a}\|_0$ means $\ell_0$ norms, i.e., the number of nonzero entries in $\mathbf{a}$.
$|\mathbf{a}|$ and $|\mathbf{a}|^2$ mean the element-wise absolute-valued vector and element-wise squared absolute-value vector for $\mathbf{a}$, respectively.
$\mathrm{Exp}(\lambda)$ is the exponential distribution with rate parameter $\lambda$, and its probability density function (PDF) is
\begin{equation}
    f_\mathrm{Exp}(x, \lambda) = \begin{cases}
        \lambda e^{-\lambda x} & x\geq 0, \\
        0 & x<0.
    \end{cases}
\end{equation}
$\underset{\mathcal{H}_0}{\overset{\mathcal{H}_1}{\gtrless}}$ represents a binary decision operator where the decision is $\mathcal{H}_0$ if the right-hand side is greater than the left-hand side, and the decision is $\mathcal{H}_1$ otherwise.
$\mathrm{diag}(\mathbf{A})$ denotes a diagonal matrix formed by the diagonal elements of the matrix $\mathbf{A}$.

\section{System Model}\label{sec:systemmodel}
We consider a single-cell massive mmWave MU-MIMO uplink system. In our model, a base station (BS) has an $M$-element uniform linear array (ULA) and serves $K$ single-antenna user equipments (UEs) via a frequency-flat channel. 
The received baseband signal vector $\mathbf{y}\in\mathbb{C}^M$ is determined as
\begin{equation}\label{eqn:yHxn}
    \mathbf{y}=\mathcal{Q}(\mathbf{Hx}+\mathbf{n}),
\end{equation}
where $\mathcal{Q}(\cdot)$ is the quantization function, $\mathbf{H}\in\mathbb{C}^{M\times K}$ is the uplink MIMO channel matrix, $\mathbf{x}\in\mathbb{C}^K$ is the transmitted signal vector, and $\mathbf{n}\sim \mathcal{CN}(0, N_0\mathbf{I}_M)$ is the additive white Gaussian noise (AWGN) vector.
The channel vector of UE $k$ is determined as
\begin{equation}\label{eqn:channelvector}
    \mathbf{h}_k = \sum_{\ell=1}^L g_{k,\ell}\mathbf{a}(\phi_{k,\ell}),
\end{equation}
where $L$ is the number of multipaths including a potential line-of-sight (LoS) propagation path, $g_{k,\ell}$ is the complex channel gain of $\ell$-th dominant path, $\phi_{k,\ell}$ is the spatial frequency which corresponds to angle-of-arrival (AoA), and $\mathbf{a}(\phi_{k,\ell})$ is steering vector. Note that $L$ is small, because the number of multipaths is limited in mmWave propagation environments due to the severe pathloss and the lack of scattering\textcolor{red}{~\cite{ref:mmwavetextbook}.}
The steering vector for arbitrary spatial frequency $\phi$ can be presented as
\begin{equation}
    \mathbf{a}(\phi) = [1, e^{-j2\pi\phi}, \cdots, e^{-j(M-1)\pi\phi}]^T.
\end{equation}
For notational simplicity, the user index is omitted in the following discussion.
Also, without loss of generality, we assume that the power of the transmitted symbols and the pilots, which are known signals for channel estimation, is 1.
However, since the modeled uplink signal in \eqref{eqn:yHxn} has a nonlinear function $\mathcal{Q}(\cdot)$, its representation can be complicated.
Thus, to simplify quantization distortion into linear form, AQNM\cite{ref:aqnm}, which is derived from Bussgang's decomposition\cite{ref:bussgang}, can be utilized.
According to this modeling, the received uplink signal can be modeled as\cite{ref:bussgang_mimo}
\begin{equation}
    \mathbf{y} = \mathcal{Q}(\mathbf{Hx}+\mathbf{n}) \approx \alpha (\mathbf{Hx+n}) + \mathbf{n}_\mathrm{q},
\end{equation}
where $\mathbf{n}_\mathrm{q}\in\mathbb{C}^M$ is the approximated quantization noise and $\alpha = \argmin_{\alpha^\prime \in \mathbb{C}} \mathbb{E}[\|\mathbf{y}-\mathbf{y}_\textnormal{orig}\|^2]$ is the quantization gain with the unquantized signal $\mathbf{y}_\textnormal{orig}$, which makes the approximated quantization noise be least correlated.
Accordingly, AQNM approximates the quantization noise vector as a zero-mean complex Gaussian noise with covariance matrix $\mathbf{R}_{\mathbf{n}_\mathrm{q}}$, which is given by\cite{ref:bussgang_mimo}
\begin{equation}
\begin{aligned}
    \mathbf{R}_{\mathbf{n}_\mathrm{q}} & = \alpha(1-\alpha) \mathrm{diag} ((\mathbf{Hx}+\mathbf{n})(\mathbf{Hx}+\mathbf{n})^H) \\
    & = \alpha(1-\alpha)\mathrm{diag}(\mathbf{Hxx}^H\mathbf{H}^H + \mathbf{nn}^H) \\
    & = \alpha(1-\alpha)\mathrm{diag}(\mathbf{Hxx}^H\mathbf{H}^H + N_0\mathbf{I}_M).
\end{aligned}
\end{equation}
Based on this assumption, the channel is observed with dedicated training phases for each UE, and the noisy observation is given by
\begin{equation}
    \mathbf{h}^\prime = \mathcal{Q}(\mathbf{h}+\mathbf{e}) \approx \alpha(\mathbf{h}+\mathbf{e}) + \mathbf{e}_\mathrm{q},
\end{equation}
where $\mathbf{e}$ is the channel estimation error due to the AWGN and $\mathbf{e}_\mathrm{q}$ is the channel estimation error from quantization noise. Note that $\mathbf{e}\sim\mathcal{CN}(0, E_0\mathbf{I}_M)$ because $\mathbf{n}$ is AWGN. Since we assumed the power of transmitted signal is 1, the error variance $E_0$ is equivalent to $N_0$. Here, based on the AQNM, the quantized signal can be simplified as
\begin{equation}
    \mathbf{h}^\prime \approx \alpha(\mathbf{h}+\mathbf{e}) + \mathbf{e}_\mathrm{q} = \alpha \mathbf{h} + \alpha\mathbf{e} + \mathbf{e}_\mathrm{q} = \alpha\mathbf{h} + \mathbf{d},
\end{equation}
where $\mathbf{d}=\alpha\mathbf{e}+\mathbf{e}_\mathrm{q}$ is the composite noise. Since $\mathbf{e}$ and $\mathbf{e}_\mathrm{q}$ is zero-mean complex Gaussian, the composite noise is also zero-mean complex Gaussian. The covariance matrix of $\mathbf{d}$ is given by
\begin{equation}
\begin{aligned}
    \mathbf{R}_\mathbf{d} & = \alpha^2\mathbf{ee}^H + \mathbf{e}_\mathrm{q}\mathbf{e}_\mathrm{q}^H \\&
    = \alpha^2E_0\mathbf{I}_M + \alpha(1-\alpha)\left(\frac{\|\mathbf{h}\|^2}{M}+E_0\right)\mathbf{I}_M
    % \\& = \left(\alpha^2\frac{\|\mathbf{h}\|^2}{M} + (\alpha^2 + 1) E_0\right)\mathbf{I}_M
    \\& = \alpha E_0 \mathbf{I}_M + \alpha(1-\alpha)\left(\frac{\|\mathbf{h}\|^2}{M}\right)\mathbf{I}_M
    \\& = D_0\mathbf{I}_M,
\end{aligned}
\end{equation}
where $D_0=\alpha E_0+\alpha(1-\alpha)({\|\mathbf{h}\|^2}/{M})$ is the composite noise variance.

For mmWave baseband processing, beamspace transformation is often employed to utilize its spatial sparsity which comes from its lack of multipaths. By this transformation, the noisy observation in the beamspace domain $\bar{\mathbf{h}}^\prime$ can be expressed as
\begin{equation}\label{eqn:beamspace}
   \bar{\mathbf{h}}^\prime=\mathbf{F\mathbf{h}}^\prime=\mathbf{F}(\alpha \mathbf{h}+\mathbf{d})
   =\alpha\mathbf{Fh}+\mathbf{Fd}=\alpha\bar{\mathbf{h}}+\bar{\mathbf{d}},
\end{equation}
where $\mathbf{F}\in\mathbb{C}^{M\times M}$ is the normalized discrete Fourier transform (DFT) matrix which implies $\mathbf{F}^H=\mathbf{F}^{-1}$, $\bar{\mathbf{h}}$ is the noiseless channel and $\bar{\mathbf{d}}$ is the estimation error induced by composite noise in beamspace representation.
Since $\mathbf{F}$ is unitary linear transformation, $\|\bar{\mathbf{h}}\|=\|\bar{\mathbf{h}}^\prime\|$, and $\bar{\mathbf{d}}$ is still zero-mean complex Gaussian.
Here, the spatial sparsity of the mmWave propagation environment is reflected in the sparsity of $\bar{\mathbf{h}}$.
Also, the Gaussianity approximation of the quantization noise holds more accurately in the beamspace domain, as shown in Fig.~\ref{fig:qnoisedistribution}, since the quantization noise is directionless.
\begin{figure}
    \centering
    \subfloat[]{\includegraphics[width=0.495\columnwidth]{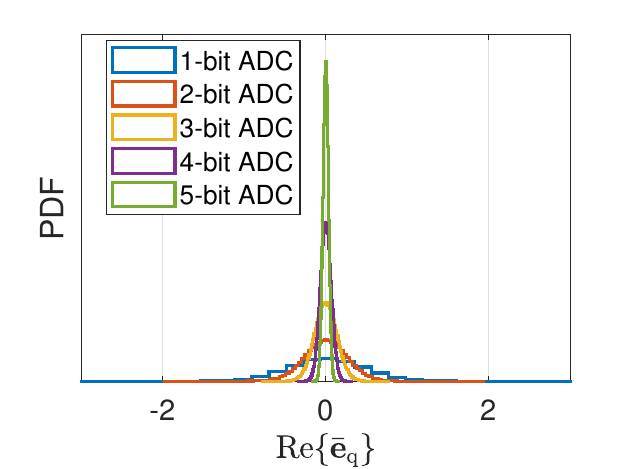}}
    \subfloat[]{\includegraphics[width=0.495\columnwidth]{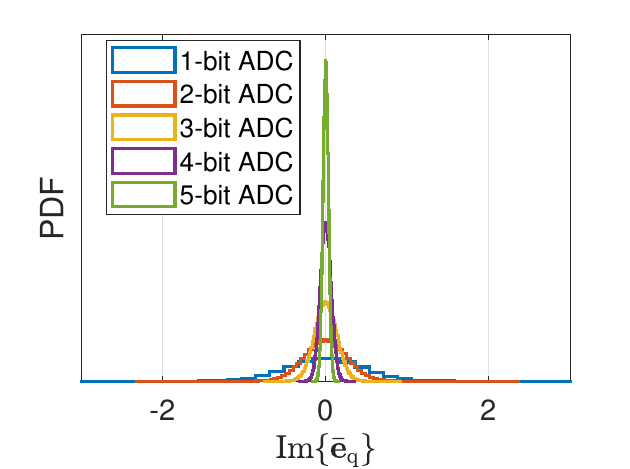}}
    \caption{Normalized histogram of the quantization noise in beamspace domain. (a) real (b) imaginary.}
    \label{fig:qnoisedistribution}
\end{figure}

\section{Proposed Method}\label{sec:proposed}
In this section, we propose a low-complexity beamspace channel denoising algorithm in the presence of quantization noise induced by low-resolution ADCs. 
First, we present a Bernoulli-complex Gaussian-based sparse channel simplification along with the corresponding problem formulation. Subsequently, we introduce auxiliary estimators required for the proposed framework. Finally, we develop a channel denoising algorithm based on a hypothesis testing-driven sparse signal denoising framework.
\subsection{Problem Formulation with Binary Hypothesis Testing}
As discussed in \eqref{eqn:channelvector} and \eqref{eqn:beamspace}, the beamspace channel vector $\bar{\mathbf{h}}$ is sparse. Therefore, recovering $\bar{\mathbf{h}}$ from $\bar{\mathbf{h}}^\prime$ can be viewed as a sparse signal denoising problem. 
To this end, we adopt a Bernoulli-complex Gaussian prior, which is broadly used in sparse signal denoising frameworks to simplify the problem and has also been widely adopted in prior works~\cite{ref:binaryht, ref:aqnm2, ref:bg1, ref:bg2, ref:mad}. 
Based on this simplification, the PDF of the $m$-th beamspace channel vector element $\bar{h}_m$ can be expressed as~\cite{ref:mad}
\begin{equation}\label{eqn:pdf_noiseless}
    f_{\bar{\mathbf{h}}}(\bar{h}_m) \triangleq q\frac{1}{\pi P_{\bar{h}}/q}
    \exp\left( - \frac{|\bar{h}_m|^2}{P_{\bar{h}}/q} \right) 
    + (1-q)\delta(\bar{h}_m),
\end{equation}
where $P_{\bar{h}}=\|\bar{\mathbf{h}}\|^2/M$ is average channel power, $q=\|\bar{\mathbf{h}}\|_0/M$ is the activity rate, and $\delta(\cdot)$ is the Dirac-delta distribution.
This indicates that the beamspace channel element $\bar{h}_m$ follows a Gaussian distribution with probability $q$, and takes the value zero with probability $(1-q)$.
Based on the equations \eqref{eqn:beamspace} and \eqref{eqn:pdf_noiseless}, the PDF of the observed noisy channel vector can be presented as~\cite{ref:binaryht}
\begin{equation}\label{eqn:pdf_noisy}
\begin{aligned}
    f_{\bar{\mathbf{h}}^\prime}(\bar{h}_m^\prime) & 
    \triangleq
    q\frac{1}{\pi(D_0+\alpha^2P_{\bar{h}}/q)} \exp \left( -\frac{|\bar{h}^\prime_{m}|^2}{D_0+\alpha^2P_{\bar{h}}/q} \right)
    \\& \quad + (1-q)\frac{1}{\pi D_0} \exp \left(-\frac{\alpha^2|\bar{h}_{m}^\prime |^2}{D_0}\right).
\end{aligned}
\end{equation}
This PDF characterizes the noisy observation by assigning each beam a probability $q$ of containing a signal component and a probability $(1-q)$ of being noise-only.
This can be rewritten as the following hypotheses:
\begin{equation}\label{eqn:hypotheses}
    \bar{h}_m^\prime = \begin{cases}
        \bar{d}_m, & \mathcal{H}_0, \\
        \alpha\bar{h}_m+\bar{d}_m, & \mathcal{H}_1, \\
    \end{cases}
\end{equation}
where $\mathcal{H}_0$ corresponds to a noise-only element, and $\mathcal{H}_1$ corresponds to an element containing both signal and noise. Under the adopted modeling, the prior probabilities of the hypotheses are equivalent to
\begin{equation}
\begin{aligned}
    &p(\mathcal{H}_0)=1-q,\\
    &p(\mathcal{H}_1)=q.
\end{aligned}
\end{equation}
These priors are derived from the Bernoulli-complex Gaussian sparsity model for the beamspace channel, and $p(\mathcal{H}_1)$ can be interpreted as the fraction of active beam elements.
Due to the inherent sparsity of mmWave channels, $q$ is typically small in practice, which justifies an asymmetric prior favoring $\mathcal{H}_0$. Based on this hypotheses, Bayesian hypothesis testing can be organized as
\begin{equation}\label{eqn:testing}
    \frac{p(\bar{\mathbf{h}}^\prime | \mathcal{H}_1)}{p(\bar{\mathbf{h}}^\prime |  \mathcal{H}_0)}
\mathop{\gtrless}_{\mathcal{H}_0}^{\mathcal{H}_1}
\frac{p(\mathcal{H}_0) C_{10}}{p(\mathcal{H}_1) C_{01}},
\end{equation}
where $C_{ij}$ denotes the penalty cost incurred when hypothesis $\mathcal{H}_i$ is selected while $\mathcal{H}_j$ is the true state.
Therefore, these costs reflect the penalties associated with misdetection (deciding $\mathcal{H}_0$ under $\mathcal{H}_1$) and false alarm (deciding $\mathcal{H}_1$ under $\mathcal{H}_0$). The ratio $C = C_{10}/C_{01}$ represents the relative penalty between the two types of errors. In particular, $C > 1$ indicates that false alarms are penalized more heavily than misdetections, and a larger value of $C$ corresponds to a stronger intolerance toward false alarms. From a denoising perspective, this implies that the probability of misclassifying noise as signal decreases, while the likelihood of treating weak signal components as noise increases. In contrast, a smaller value of $C$ leads to more weak signal components being classified as signal, at the expense of also increasing the probability of misclassifying noise as signal.
Using the PDF in \eqref{eqn:pdf_noisy} and hypotheses in \eqref{eqn:hypotheses}, the testing in \eqref{eqn:testing} can be organized as
\begin{equation}\label{eqn:testing_orig}
    \dfrac{\frac{1}{\pi(\alpha^2P_{\bar{h}}/p(\mathcal{H}_1)+D_0)} \exp(-\frac{|\bar{h}^\prime_m|^2}{\alpha^2P_{\bar{h}}/p(\mathcal{H}_1)+D_0})}
    {\frac{1}{\pi D_0}\exp(-\frac{|\bar{h}_m^\prime|^2}{D_0})}
    \mathop{\gtrless}_{\mathcal{H}_0}^{\mathcal{H}_1}
    \frac{p(\mathcal{H}_0)}{p(\mathcal{H}_1)}C.
\end{equation}
This testing can be reformulated in terms of $|\bar{h}_m^\prime|^2$ as follows:
\begin{equation}\label{eqn:threshold}
    |\bar{h}_m^\prime|^2 \mathop{\gtrless}_{\mathcal{H}_0}^{\mathcal{H}_1}
    D_0 \left(1+\frac{p(\mathcal{H}_1)}{\mathrm{SDNR}}\right)\ln \left( \left(1 + \frac{\mathrm{SDNR}}{p(\mathcal{H}_1)}\right) \frac{1-p(\mathcal{H}_1)}{p(\mathcal{H}_1)}C  \right),
\end{equation}
where $\mathrm{SDNR}$ is the signal-to-distortion-plus-noise power ratio (SDNR), which is defined as
\begin{equation}
    \mathrm{SDNR} = \frac{\alpha^2P_{\bar{h}}}{D_0}.
\end{equation}
For the full derivation, please refer to Appendix~\ref{app:testderivation}.
This threshold increases with the composite noise variance and the signal sparsity, and tends to classify elements that are more likely to be inactive as noise.
\begin{remark}
    Bernoulli-complex Gaussian modeling does not strictly represent the characteristics of mmWave channel propagation; rather, it is introduced as a simplified representation to reformulate the channel denoising problem as a sparse signal denoising task. Furthermore, although quantization noise is not inherently Gaussian, it can be roughly approximated as Gaussian in the beamspace domain due to its non-directional nature. Based on these considerations, the adopted problem formulation serves as a surrogate model that captures both the mmWave channel and nonlinear quantization noise, while significantly reducing computational complexity by simplifying the denoising problem.
\end{remark}

\subsection{Auxiliary Estimators}
The above testing procedure has a limitation in that it requires several parameters that are not available a priori without channel estimation, such as $D_0$, the SDNR, and $p(\mathcal{H}_1)$. 
To address this issue, we employ blind estimators that enable parameter estimation without relying on pilot signals and channel estimates. To estimate the composite noise variance $D_0$, we adopt a robust statistical approach. Since the signal components are sparse, most entries are close to zero. Consequently, from the perspective of the noisy observations, the majority of samples follow a dominant distribution corresponding to the composite noise, while a small subset of samples that contain signal components with large magnitudes follow a different distribution and can be regarded as statistical outliers. Therefore, the problem of estimating the composite noise variance can be interpreted as a robust variance estimation problem in the presence of outliers. Based on this interpretation, we propose a blind estimator for the composite noise power.
Since the composite noise $\bar{\mathbf{d}}$ follows a complex Gaussian distribution, the element-wise squared magnitude vector $|\bar{\mathbf{d}}|^2$ follows an exponential distribution, i.e., $|\bar{\mathbf{d}}|^2 \sim \mathrm{Exp}(D_0)$. 
From the characteristic of the exponential distribution, the median of this distribution is given by $D_0 \ln 2$. 
By matching the scale with respect to $D_0$, an initial estimate of the variance can be derived as~\cite{ref:mad}
\begin{equation} 
    \widehat{D}_0^{(0)} = \frac{\mathrm{median}(|\bar{\mathbf{h}}^\prime|)}{\ln 2}, 
\end{equation} 
where $\mathrm{median}(\cdot)$ is median value of the vector.
However, this estimate is suitable only under the assumption of strong sparsity, whereas beamspace channels may not be sufficiently sparse for the median absolute deviation (MAD) estimator to operate reliably. 
Furthermore, it has been observed that, in such environments, the estimator exhibits a positive bias~\cite{ref:mad}.
To refine the estimate, we employ a truncated mean that suppresses the influence of outliers. In robust statistics, the widely used Huber-M estimator~\cite{ref:huber} achieves robust variance estimation by reducing the contribution of outliers. Specifically, elements identified as outliers are downweighted and contribute linearly rather than quadratically in the variance computation, while inlier elements are incorporated without modification.
However, in the proposed method, a more aggressive strategy is adopted to mitigate the bias induced by signal components. 
Specifically, elements identified as outliers are completely truncated to zero, and the resulting estimate is refined through a small number of iterations. 
This procedure can be formulated as follows. 
At the $t$-th iteration, outliers are identified based on a confidence interval derived from the noise variance estimate obtained in the previous step. 
The corresponding threshold is given by
\begin{equation}
    \tau^{(t)} = c\widehat{D}_0^{(t-1)},
\end{equation}
where $c$ is a parameter determined by the confidence level, implying that the probability that a noise-only element is smaller than $\tau^{(t)}$ is $(1 - e^{-c})$. 
Accordingly, the set of noise-dominated elements is defined as
\begin{equation}\label{eqn:setofnoise}
    \mathcal{S}^{(t)} = \{m: |\bar{h}^\prime_m|^2 \leq \tau^{(t)}\}.
\end{equation}
If the number of retained samples is too small, the resulting average may become unreliable. 
Therefore, when the cardinality satisfies $|\mathcal{S}^{(t)}| < \rho_{\mathrm{min}}$, where $\rho_{\mathrm{min}}$ denotes the minimum required number of samples, the set is reconstructed with $c'>c$ to modify the confidence level so that more samples are included.
Based on this set, the variance is estimated using the truncated samples, which inherently introduces bias. As discussed above, since $|\bar{\mathbf{d}}|^2 \sim \mathrm{Exp}(D_0)$, the truncated mean obtained by retaining only the samples below the threshold $c D_0$ is smaller than the true mean $D_0$. Specifically, Lemma~\ref{lem:truncatedmeanbias} quantifies the bias induced by truncation.
\begin{lemma}\label{lem:truncatedmeanbias}
    Let $X\sim\mathrm{Exp}(\lambda)$. Then, for a truncation threshold $k\lambda$,
    \begin{equation}
        \mathbb{E}[X|X\leq k\lambda] = \lambda \kappa(k),
    \end{equation}
    where
    \begin{equation}
        \kappa(k) = \frac{1-e^{-k}(1+k)}{1-e^{-k}}.
    \end{equation}
\end{lemma}
\begin{proof}
    Please refer to Appendix~\ref{app:prooflem}.
\end{proof}
Thus, the truncated mean exhibits a negative bias of $\kappa(c)$. Accordingly, a bias correction is applied to obtain the noise variance estimate as
\begin{equation}\label{eqn:noisepowersum}
    \widehat{D}_0^{(t+1)} = \frac{1}{\kappa(c)}\cdot \frac{1}{|\mathcal{S}^{(t)}|}
    \sum_{m\in\mathcal{S}^{(t)}} |\bar{h}^\prime_m|^2,
\end{equation}
where $1/\kappa(c)$ is a correction factor that accounts for the bias induced by truncation.
Since this estimate may still be unstable, particularly in the early iterations, it is further refined through a small number of iterations $T$.

The overall procedure of the composite noise variance estimation is summarized in Algorithm~\ref{alg:noisepwrest}.
The computational complexity of this procedure can be analyzed as follows. 
First, finding the median in the MAD initialization can be performed with a complexity of $\mathcal{O}(M)$ in both the average and worst cases by using the median-of-medians algorithm~\cite{ref:medianofmedians}. 
The formation of the noise-dominated set in \eqref{eqn:setofnoise} requires $\mathcal{O}(M)$. 
If the cardinality of the set is smaller than the minimum required number of samples $\rho_{\mathrm{min}}$, the adjustment process also incurs a complexity of $\mathcal{O}(M)$.
The subsequent summation and bias correction in \eqref{eqn:noisepowersum} also require $\mathcal{O}(M)$. 
Since this procedure is repeated for $T$ iterations, the overall computational complexity is $\mathcal{O}(TM)$.

\begin{algorithm}[t]
\caption{Blind Composite Noise Power Estimator}\label{alg:noisepwrest}
\SetAlgoLined
\textbf{input:} $|\bar{\mathbf{h}}^\prime|^2$, $c$, $c^\prime$, $\rho_\mathrm{min}$, $T$

\textbf{initialize} $\widehat{D}_0^{(0)}=\dfrac{\mathrm{median}(|\bar{\mathbf{h}}^\prime|^2)}{\ln 2}$

\For{$t=1,...,T$}{ 
    Calculate the threshold
    \begin{equation}
        \tau^{(t)} = c\widehat{D}_0^{(t-1)}.
    \end{equation}

    Determine the set of noise as
    \begin{equation}\label{eqn:setofnoise1}
        \mathcal{S}^{(t)} = \{m: |\bar{h}^\prime_m|^2\leq \tau^{(t)}\}.
    \end{equation}
    
    \If{$|\mathcal{S}^{(t)}|< \rho_\mathrm{min}$}{
        Increase the confidence interval and determine the set of noise again as
        \begin{equation}\label{eqn:setofnoise2}
            \mathcal{S}^{(t)} = \{m: |\bar{h}^\prime_m|^2\leq c^\prime\widehat{D}_0^{(t-1)}\}.
        \end{equation}
    }
    Calculate the noise power estimate as
    \begin{equation}
        \widehat{D}_0^{(t)} = \frac{1}{\kappa(c)}\cdot\frac{1}{|\mathcal{S}^{(t)}|}\sum_{m\in\mathcal{S}^{(t)}}|\bar{h}_m^\prime|^2.
    \end{equation}
    }
    
\Return {Composite noise power estimate $\widehat{D}_0$}    

\end{algorithm}

Additionally, under the assumption that the signal and the composite noise are uncorrelated under the AQNM, the following relationship holds:
\begin{equation}
\mathbb{E}[|\bar{\mathbf{h}}^\prime|^2] = \mathbb{E}[\alpha |\bar{\mathbf{h}}|^2] + \mathbb{E}[|\bar{\mathbf{d}}|^2].
\end{equation}
Accordingly, the average channel power can be derived as~\cite{ref:mad, ref:binaryht}
\begin{equation}\label{eqn:channelpwr}
    \widehat{P}_{\bar{h}} = \max\left\{ \frac{\|\bar{\mathbf{h}}^\prime\|^2}{M} - \widehat{D}_0 \right\},
\end{equation}
and the SDNR can be derived as follows~\cite{ref:mad, ref:binaryht}:
\begin{equation}\label{eqn:snrest}
\widehat{\mathrm{SDNR}} = \max\left\{ \frac{\|\bar{\mathbf{h}}^\prime\|^2}{M\widehat{D}_0} - 1, 0 \right\}.
\end{equation}
The clipping at zero is introduced to prevent negative values in the channel power and linear-scale SDNR, which may occur when the noise power is overestimated due to estimation inaccuracies.
Note that this estimator has a computational complexity of $\mathcal{O}(M)$.

Using these estimates, the activity rate $q=p(\mathcal{H}_1)$ can be estimated as\cite{ref:binaryht}
\begin{equation}\label{eqn:activityest}
    \hat{q} = \argmin_{q^\prime = \{\frac{m}{M} | m=1,\cdots,M\}} \left| \dfrac{2(\widehat{\mathrm{SDNR}})^2}{\frac{1}{M\widehat{D}_0^2}\sum_{m=1}^M|\bar{h}_m^\prime|^4  -2-4\widehat{\mathrm{SDNR}}} - q^\prime \right|
\end{equation}
Note that this estimator also has a computational complexity of $\mathcal{O}(M)$.

\subsection{Denoising}
Based on the above estimates, the threshold in \eqref{eqn:threshold} is computed to distinguish between elements that contain only noise and those that contain both signal and noise, and it can be presented as
\begin{equation}\label{eqn:realthreshold}
    |\bar{h}_m^\prime|^2 \mathop{\gtrless}_{\mathcal{H}_0}^{\mathcal{H}_1}
    \widehat{D}_0 \left(1+\frac{\hat{q}}{\widehat{\mathrm{SDNR}}}\right)\ln \left( \left(1 + \frac{\widehat{\mathrm{SDNR}}}{\hat{q}}\right) \frac{1-\hat{q}}{\hat{q}}C  \right).
\end{equation}
Accordingly, denoising is performed by zeroing the elements that are more likely to contain only noise and scaling to compensate for the quantization gain $\alpha$, which can be expressed as
\begin{equation}\label{eqn:denoising}
    \bar{h}_m^\star =
    \begin{cases}
        \dfrac{1}{\alpha}\bar{h}_m^\prime & \textnormal{if } \mathcal{H}_1, \\
        0 & \textnormal{if } \mathcal{H}_0.
    \end{cases}
\end{equation}
This operation can be interpreted as a hard-thresholding rule, where elements exceeding the threshold are retained while the others are set to zero. 
Additionally, since the quantization gain is applied to the noisy observation as shown in \eqref{eqn:beamspace}, it is compensated by scaling with $1/\alpha$. Here, the value of $\alpha$ is determined following the derivations in \cite{ref:aqnm_values} according to the ADC resolution~\cite{ref:aqnm}.
Since this process involves a single pass over all elements, the overall computational complexity is $\mathcal{O}(M)$.

Overall, the proposed method performs channel denoising through five main steps, as outlined in Algorithm~\ref{alg:overall}. 
The first step is blind composite noise variance estimation, described in Algorithm~\ref {alg:noisepwrest}. 
The second step is blind channel power and SDNR estimation in \eqref{eqn:channelpwr} and \eqref{eqn:snrest}, followed by blind activity rate estimation in~\eqref{eqn:activityest}. 
The fourth step computes the threshold as in~\eqref{eqn:realthreshold}, 
and the final step performs denoising in~\eqref{eqn:denoising}.
As discussed above, the computational complexity of the first step is $\mathcal{O}(TM)$, while the second, third, and fifth steps each have a complexity of $\mathcal{O}(M)$. The threshold computation in the fourth step is based on closed-form simple arithmetic and has a complexity of $\mathcal{O}(1)$. Therefore, the overall computational complexity of the proposed method is $\mathcal{O}(TM)$. Since we assume that $T$ is small, the proposed denoiser has a near-linear computational complexity.

\begin{remark}\label{rmk:snrknowledge}
    If the system has prior knowledge of the AWGN variance or SNR, the composite noise variance $D_0$ can be determined based on commonly used values in the AQNM. 
    In this case, the computation required for estimating the composite noise power, channel power, and the SDNR can be omitted, thereby reducing the computational complexity of the denoising process to $\mathcal{O}(M)$.
\end{remark}

\begin{algorithm}[t]
\caption{Low-Complexity Beamspace Channel Denoising}\label{alg:overall}
\SetAlgoLined
\textbf{input:} $\bar{\mathbf{h}}^\prime$, $C$, $c$, $c^\prime$, $\rho_\mathrm{min}$, $T$\\
\textbf{initialization:} denoised beamspace channel vector $\bar{\mathbf{h}}^\star=\mathbf{0}^{M\times 1}$\\
Calculate the composite noise power variance estimate $\widehat{D}_0$ with Algorithm~\ref{alg:noisepwrest}. \\
Estimate $\widehat{\mathrm{SDNR}}$ and $\hat{q}$ using $\widehat{D}_0$.\\
Calculate the testing threshold with $\widehat{D}_0,\widehat{\mathrm{SDNR}},\hat{q}$, which is given by
\begin{equation}
    \eta = \widehat{D}_0\left( \frac{\hat{q}}{\widehat{\mathrm{SDNR}}}+1 \right) \ln \left(\left( 1+\frac{\widehat{\mathrm{SDNR}}}{\hat{q}} \right) \frac{1-\hat{q}}{\hat{q}}C \right).
\end{equation}\\
\For{$m=1:M$}{
    \If{$|\bar{h}^\prime_m|^2\geq \eta$}{
    $\bar{h}_m^\star = \dfrac{1}{\alpha} \bar{h}_m^\prime$.
    }
}
\Return{\textnormal{denoised beamspace channel vector} $\bar{\mathbf{h}}^\star$}    
\end{algorithm}
\section{Simulation Results}\label{sec:simresults}
\subsection{Simulation Configuration}
For performance evaluation, we consider a BS equipped with a 64-element ULA, where the antenna spacing is set to half of the carrier wavelength. For the channel model, we employ the QuaDRiGa mmMAGIC UMi model~\cite{ref:quadriga} with a carrier frequency of 50 GHz. The ADC quantization levels are configured to minimize the quantization noise, and all antennas share the same quantization resolution. 
For input parameters, we set $c$ and $c^\prime$ to 2 and 4, respectively, to facilitate hardware implementation, since powers of two can be realized using shift operations instead of multiplications.
These values correspond to probabilities of 86.5\% and 98.2\%, respectively, for a noise-only sample to fall below the threshold. In addition, the cost parameter $C$ is also set to 4, which is favorable for hardware implementation since it is a power-of-two value.

We perform 10000 Monte Carlo trials, where the channel is independently regenerated for each trial. To evaluate the channel estimation performance of the proposed method, we compare the mean-squared error (MSE) with several baseline algorithms. As basic baselines, we consider the least-squares (LS)-based estimator, which uses the noisy observation directly as the channel estimate and the Bussgang-based linear MMSE (LMMSE) channel estimation method~\cite{ref:aqnm_extend}. In addition, we include comparisons with $\alpha$-BEACHES~\cite{ref:sand}, which is a low-complexity beamspace channel denoising algorithm, GL-QVBCE~\cite{ref:glqvbce}, which is one of the state-of-the-art CS-based methods, LDGEC~\cite{ref:ldgec}, which is the deep learning-based beamspace channel denoising algorithm, and a diffusion model-based approach~\cite{ref:ul1}. 
Furthermore, to assess the practical impact of channel estimation accuracy, we evaluate the uncoded bit-error rate (BER) for uplink transmission with 8 users and 16-quadrature amplitude modulation (16QAM). Specifically, an LMMSE equalizer is applied, where the filter matrix is given by
\begin{equation}
    \mathbf{W} = \hat{\mathbf{H}}( \hat{\mathbf{H}}^H \hat{\mathbf{H}}+\frac{1}{\mathrm{SNR}}\mathbf{I}_K)^{-1},
\end{equation}
where $\hat{\mathbf{H}}$ is the estimated channel matrix.
To provide a performance benchmark, we also include the case where the system has access to perfect CSI, allowing us to compare the performance with the optimal scenario.

\subsection{Complexity Analysis}
We compare the computational complexity of the proposed method with that of several baseline estimators. 
Table~\ref{tbl:complexity_time} summarizes the complexity of each channel estimation technique in big-$\mathcal{O}$ notation, along with the normalized operation time measured in Python.
The proposed method requires $\mathcal{O}(M \log M)$ complexity for the beamspace transformation and $\mathcal{O}(TM)$ for the denoising process, resulting in an overall low computational complexity compared to other methods.
The Bussgang-based LMMSE estimator~\cite{ref:aqnm_extend} is the simplest approach, with $\mathcal{O}(M)$ complexity. However, it requires prior knowledge of the SNR. If this assumption is incorporated, as discussed in Remark~\ref{rmk:snrknowledge}, the overall complexity of the proposed method becomes $\mathcal{O}(M \log M + M) = \mathcal{O}(M \log M)$.
$\alpha$-BEACHES~\cite{ref:sand} has $\mathcal{O}(M \log M)$ complexity, since both the beamspace transformation and the denoising parameter search incur $\mathcal{O}(M \log M)$ complexity. This method also requires prior knowledge of the SNR. Under the same assumption that SNR is known and the beamspace transformation is equally performed, the proposed method requires only $\mathcal{O}(M)$ for denoising, whereas $\alpha$-BEACHES requires $\mathcal{O}(M \log M)$ for denoising, indicating that the proposed method is computationally more efficient.
GL-QVBCE~\cite{ref:glqvbce} has a complexity of $\mathcal{O}(N_I(\hat{L}^3 + M\hat{L}^2))$, where $N_I$ denotes the number of iterations and $\hat{L}$ is the estimated number of dominant propagation paths. Although mmWave channels are sparse and $\hat{L}$ is typically small, the normalized operation time is relatively high due to the large number of iterations required to satisfy the convergence criteria.
The LDGEC~\cite{ref:ldgec} requires $\mathcal{O}(N_I M^3)$ complexity, as it involves matrix inversion at each iteration. This leads to a relatively long operation time compared to other methods. Additionally, since it is also a beamspace channel estimation algorithm, it requires beamspace transformation for preprocessing, which incurs a complexity of $\mathcal{O}(M \log M)$. However, it is dominated by the $\mathcal{O}(N_I M^3)$ term and is therefore omitted in the big-O notation.
The diffusion-based approach~\cite{ref:ul1} has a complexity of $\mathcal{O}(M \log M + N_I(N_L M k C^2 + M^2))$, where $N_L$, $k$, and $C$ denote the number of CNN layers, kernel size, and channel depth, respectively. This is because the method performs iterative posterior updates in a reverse diffusion process, combining CNN-based prior information and likelihood gradients on the beamspace-transformed channel. Although its normalized operation time is slightly shorter than some iterative methods, due to the absence of highly complex operations within each iteration, it is still significantly higher than that of the proposed method.

% \begin{table*}[]
%     \caption{Computational Complexity and Normalized Operation Time}\label{tbl:complexity_time}
%     {\centering
%     \begin{tabular}{|c|ccc|}
%     \hline
%     Algorithm  & Complexity                   & Runtime & Learning-based \\ \hline\hline
%     Bussgang-based LMMSE$^a$ & $\mathcal{O}(M)$                &                           & No             \\ 
%     $\alpha$-BEACHES$^a$ & $\mathcal{O}(M\log M)$ & & No \\
%     GL-QVBCE   & $\mathcal{O}(N_I(\hat{L}^3 + M\hat{L}^2))$             &                           & No             \\ 
%     LDGEC      & $\mathcal{O}(N_IM^3)$          &                           & Yes            \\ 
%     Diffusion  & $\mathcal{O}(M\log M +N_I(N_LMkC^2+M^2))$              &                           & Yes            \\ 
%     Proposed   & $\mathcal{O}(M \log M + TM)$ & 1.0                       & No             \\ \hline
%     \end{tabular}\par}
%     \vspace{2mm}
%     {\footnotesize
%     $^a$Requires the knowledge of SNR}
% \end{table*}

\begin{table*}[]
\caption{Computational Complexity and Normalized Operation Time}\label{tbl:complexity_time}
    {\centering
\begin{tabular}{|c|c|c|c|c|}
\hline
Algorithm                & Complexity                                 & \multicolumn{1}{c|}{\begin{tabular}[c]{@{}c@{}}Runtime \\ (with prior knowledge)\end{tabular}} & \begin{tabular}[c]{@{}c@{}}Runtime\\ (without prior knowledge)\end{tabular} & Learning-based \\ \hline\hline
Bussgang-based LMMSE$^a$ & $\mathcal{O}(M)$                           & \multicolumn{1}{c|}{0.05$\times$}                                                                                 & -                                                                                  & No             \\ \hline
$\alpha$-BEACHES$^a$     & $\mathcal{O}(M \log M)$                    & \multicolumn{1}{c|}{1.39$\times$}  & - & No \\ \hline
GL-QVBCE & $\mathcal{O}(N_I(\hat{L}^3 + M\hat{L}^2))$ & \multicolumn{2}{c|}{622$\times$}                                                                    & No             \\ \hline
LDGEC                    & $\mathcal{O}(N_IM^3)$                      & \multicolumn{2}{c|}{1453$\times$}                                                                                                                                                                      & Yes            \\ \hline
Diffusion                & $\mathcal{O}(M\log M +N_I(N_LMkC^2+M^2))$  & \multicolumn{2}{c|}{514$\times$}                                                                                                                                                                      & Yes            \\ \hline
Proposed                 & $\mathcal{O}(M \log M + TM)$               & \multicolumn{1}{c|}{1.00$\times$}                                                                              & 1.24$\times$  & No             \\ \hline
\end{tabular}\par}
    \vspace{2mm}
    {\footnotesize
    $^a$Requires the knowledge of SNR}
\end{table*}

\begin{figure}
    \centering
    \includegraphics[width=\linewidth]{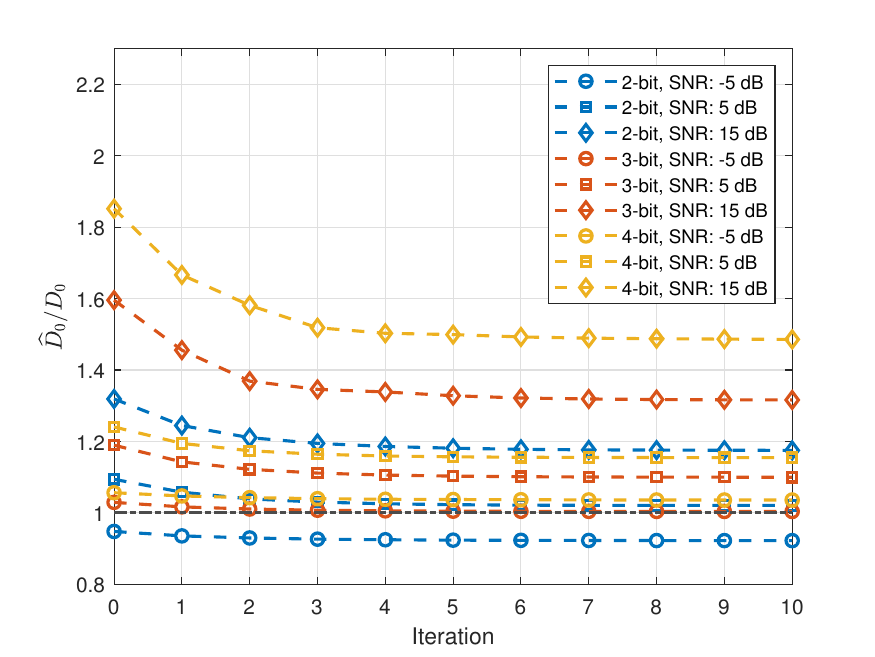}
    \caption{Estimated $D_0$ versus iterations.}
    \label{fig:iteration_D0}
\end{figure}

\subsection{Numerical Results and Analysis}
Fig.~\ref{fig:iteration_D0} shows the average $\widehat{D}_0$ across iterations, normalized with respect to the ground-truth $D_0$. The value at iteration 0 corresponds to the initialization obtained using the MAD estimator, and the estimator tends to slightly overestimate the variance as the SNR increases. This is because the initial value from the MAD estimate becomes larger, resulting in a wider confidence interval, which in turn causes more beamspace elements containing both signal and noise components to be misclassified as noise. Furthermore, it tends to yield higher estimation accuracy for lower ADC resolutions. This is because the impact of quantization noise becomes more significant, increasing the power of the composite noise. As a result, even if some beamspace elements that are not noise-dominant are misclassified as noise, the estimation becomes less sensitive to such misclassifications. 
In addition, the estimation generally converges within three iterations. Based on this observation, the number of iterations is set to three in both following simulations and hardware implementations.

\begin{figure}
    \centering
    \subfloat[MSE]{\includegraphics[width=\linewidth]{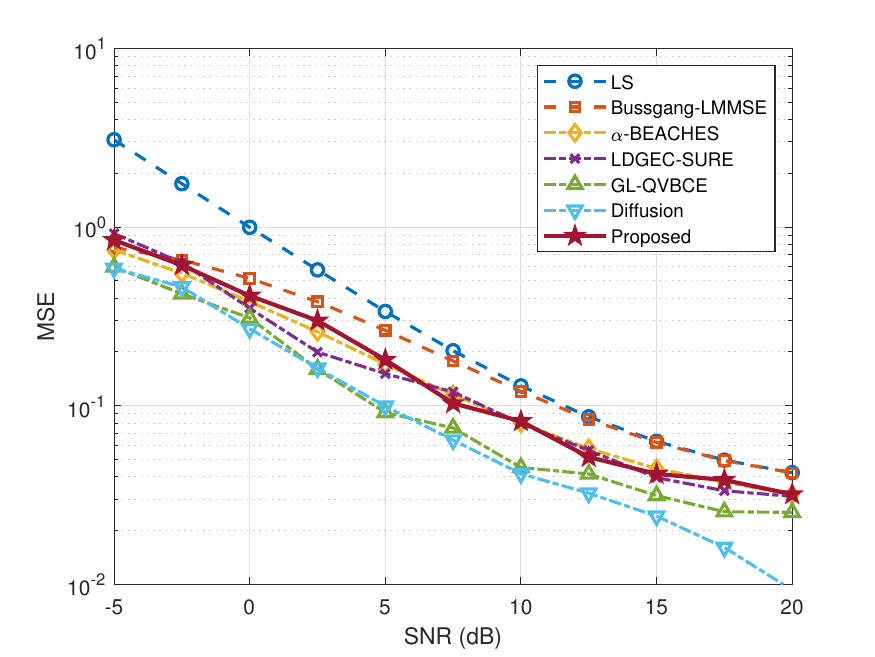}\label{subfig:mse_comparison}}

    \subfloat[Post-equalization BER]{\includegraphics[width=\linewidth]{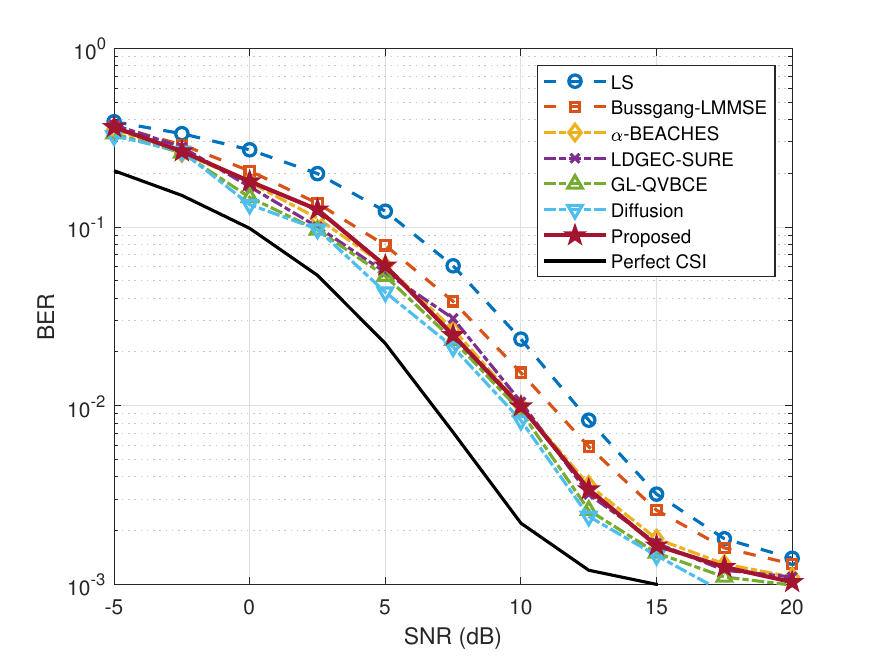}\label{subfig:ber_comparison}}
    \caption{Performance depending on SNR compared to baselines with 3-bit ADCs.}
    \label{fig:comparison}
\end{figure}

Fig.~\ref{fig:comparison} presents the channel estimation performance under a 3-bit ADC scenario. As shown in Fig.~\ref{subfig:mse_comparison}, the proposed method achieves comparable channel estimation performance to LDGEC-SURE and $\alpha$-BEACHES, while outperforming Bussgang-LMMSE and LS estimators. On the other hand, it exhibits slightly higher MSE compared to GL-QVBCE and diffusion-based methods, which can be attributed to the performance trade-off resulting from reduced computational complexity. 
Nevertheless, as shown in Fig.~\ref{subfig:ber_comparison}, the post-equalization BER of the proposed method is nearly identical to that of computationally intensive approaches. This indicates that the proposed method can achieve practically comparable performance to methods relying on deep learning, iterative optimization, or parameter search, despite not employing such techniques.

Fig.~\ref{fig:bit} illustrates the performance gain achieved by applying the proposed denoising method by comparing it with the LS estimate, i.e., a noisy observation before denoising.
As shown in Fig.~\ref{subfig:mse_bit}, the proposed method significantly reduces the MSE across all SNR regimes and all ADC resolutions compared to the case without denoising. Moreover, the performance gain becomes more pronounced as the ADC resolution decreases, and this trend is particularly evident in the low-SNR regime. In terms of post-equalization BER, as shown in Fig.~\ref{subfig:ber_bit} for 2-3-bit ADCs, the proposed denoising provides approximately 3-5 dB SNR gain at the same BER. Meanwhile, for 4-bit ADCs, the performance gain is around 1 dB at the same BER level. This is because the impact of quantization noise to be suppressed by denoising becomes less significant as the ADC resolution increases. These results demonstrate that the proposed method provides meaningful performance improvements despite its low computational complexity.

\begin{figure}
    \centering
    \subfloat[MSE]{\includegraphics[width=\linewidth]{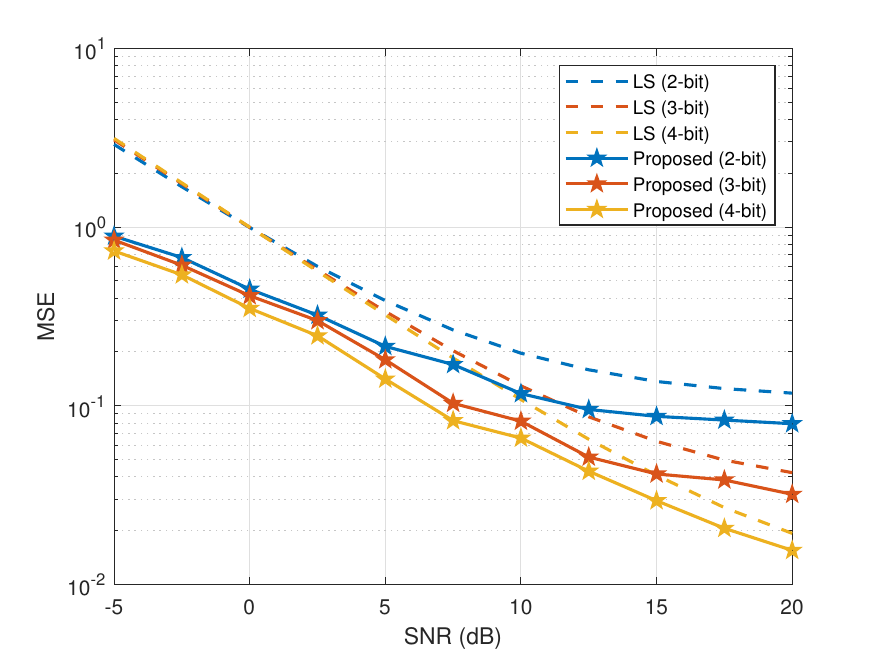}\label{subfig:mse_bit}}

    \subfloat[Post-equalization BER]{\includegraphics[width=\linewidth]{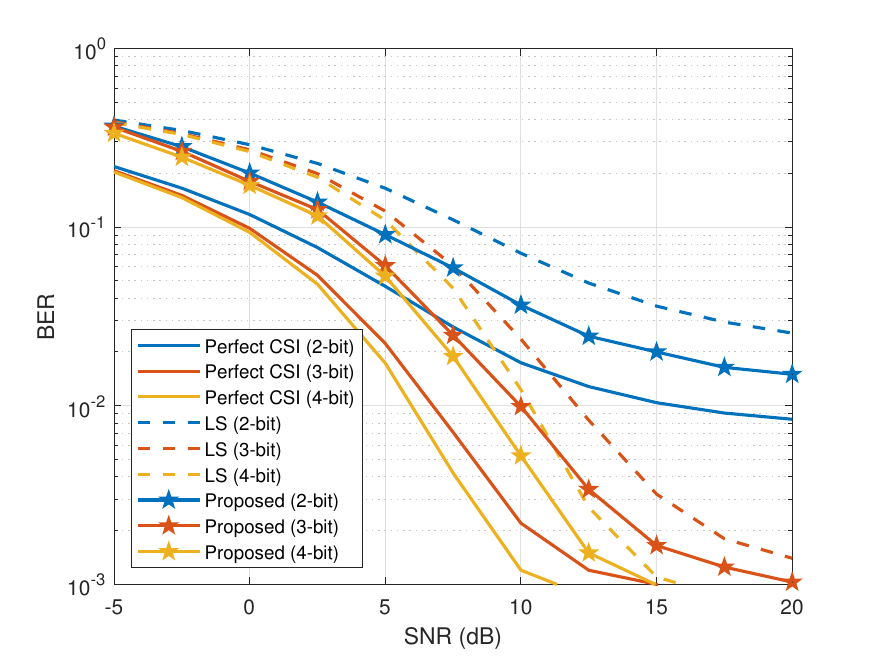}\label{subfig:ber_bit}}
    \caption{Performance depending on SNR and ADC resolutions.}
    \label{fig:bit}
\end{figure}

\section{Hardware Design and Implementation}\label{sec:hardware}
In this section, we present a VLSI architecture for the proposed channel estimation algorithm devised in Section~\ref{sec:proposed}. We also provide the corresponding FPGA implementation results.

\subsection{Hardware Architecture}
Fig.~\ref{fig:vlsi_highlevelarchi} shows a high-level block diagram of the proposed VLSI architecture, which implements the algorithm devised in the previous section. 
The architecture consists of four main parts: (i) preprocessing and postprocessing units, including the fast Fourier transform (FFT) unit, element-wise magnitude squaring, and inverse FFT (IFFT), (ii) a parameter estimation unit consisting of the composite noise variance estimator, channel power estimator, SDNR estimator, and activity rate estimator, (iii) a threshold calculation unit, and (iv) a denoising unit. 
In the preprocessing stage, the FFT unit transforms the received signal from the antenna domain to the beamspace domain. 
The transformed signal is then fed into the element-wise magnitude squaring unit, where it is converted into squared magnitudes, while the original beamspace signal is also used as the noisy observation for the denoising process.
The squared magnitudes are utilized in the auxiliary estimation unit to estimate the required parameters, and are additionally employed in the hypothesis testing procedure within the denoising unit.
Note that the units in the yellow area can be omitted if prior knowledge of the SNR is available, as discussed in Remark~\ref{rmk:snrknowledge}.
The estimated parameters are subsequently used in the threshold calculation unit to determine the threshold value $\eta$.
Based on $\eta$ and the element-wise squared magnitudes, each beamspace channel element is processed by the denoising unit, as shown in Algorithm~\ref{alg:overall}.
Finally, the denoised beamspace channel vector is transformed back to the antenna domain via the IFFT.

To improve hardware efficiency, a two’s complement fixed-point representation is employed. 
The word lengths are optimized to maintain performance comparable to floating-point operation results. 
The antenna domain signals are represented with 16 bits and 8 fractional bits, whereas the beamspace domain signals use 10 bits with 8 fractional bits, as the scaling is introduced by the FFT operation. The output of the element-wise squaring unit is expressed using 16 bits with 8 fractional bits, and the same representation is applied to both $\widehat{D}_0$ and $\widehat{P}_{\bar{h}}$. In addition, the SDNR estimate is represented using 24 bits with 8 fractional bits.

\begin{figure*}
    \centering
    \includegraphics[width=0.9\linewidth]{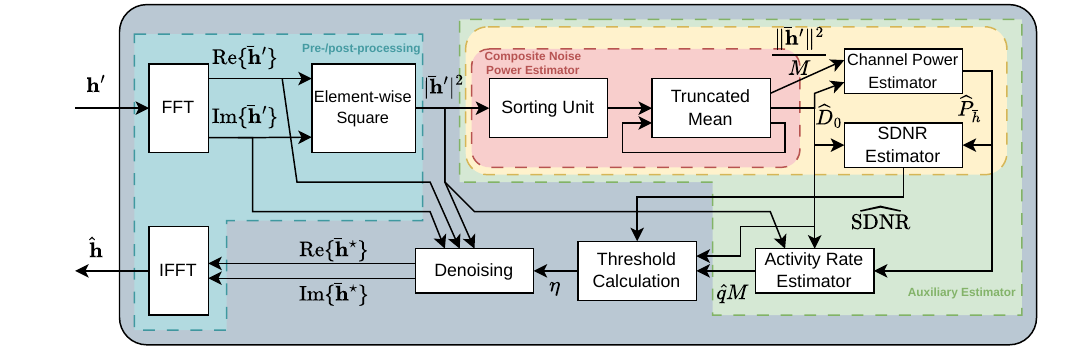}
    \caption{High-level architecture of the proposed algorithm in VLSI design; the yellow area (estimators for composite noise power, channel power, and SDNR) can be omitted if the prior knowledge of the SNR is available.}
    \label{fig:vlsi_highlevelarchi}
\end{figure*}

Fig.~\ref{fig:vlsi_d0} illustrates the proposed $\widehat{D}_0$ estimation unit, which consists of a sorting unit and a truncated mean unit. The sorting unit outputs one sorted value per clock cycle using an insertion-sort-based mechanism, while the truncated mean unit estimates the composite noise power $\widehat{D}_0$ based on the sorted sequence.
Although the sorting operation is not explicitly included in Algorithm~\ref{alg:noisepwrest}, it is incorporated in the hardware architecture to reduce computational complexity. Specifically, the sorted values can be directly utilized for median computation and facilitate efficient evaluation of \eqref{eqn:setofnoise} and the corresponding truncated mean, thereby lowering the overall implementation cost.
As shown in Fig.~\ref{subfig:sortingunit}, the sorting unit is implemented using a systolic insertion-like architecture, where sorting is progressively achieved as data propagate through a cascade of pipeline registers.
The module is controlled by a finite-state machine (FSM) that manages three phases: data loading, pipeline flushing, and sorted output generation. During the loading phase, input samples are streamed into the pipeline at a rate of one sample per clock cycle. Meanwhile, each stage continuously performs a local compare-and-swap operation, where the smaller value is retained locally and the larger value is forwarded to the next stage.
After all inputs are loaded, a flushing phase of $(M-1)$ cycles ensures that remaining values fully propagate through the pipeline. Finally, the sorted outputs are sequentially emitted.
Since each stage performs only a single comparison per cycle and the datapath operates in a fully pipelined manner without global control fanout, the architecture avoids long combinational paths and achieves timing-efficient hardware implementation.

\begin{figure}
    \centering
    \subfloat[]{\includegraphics[width=\linewidth, trim=40 0 40 0, clip]{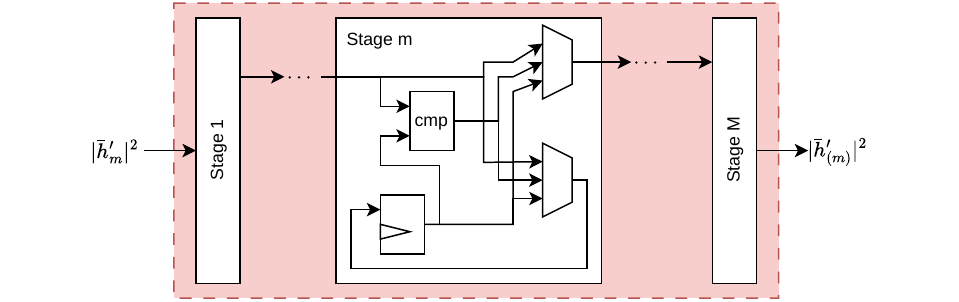}\label{subfig:sortingunit}}
    
    \subfloat[]{\includegraphics[width=\linewidth, trim=40 0 40 0, clip]{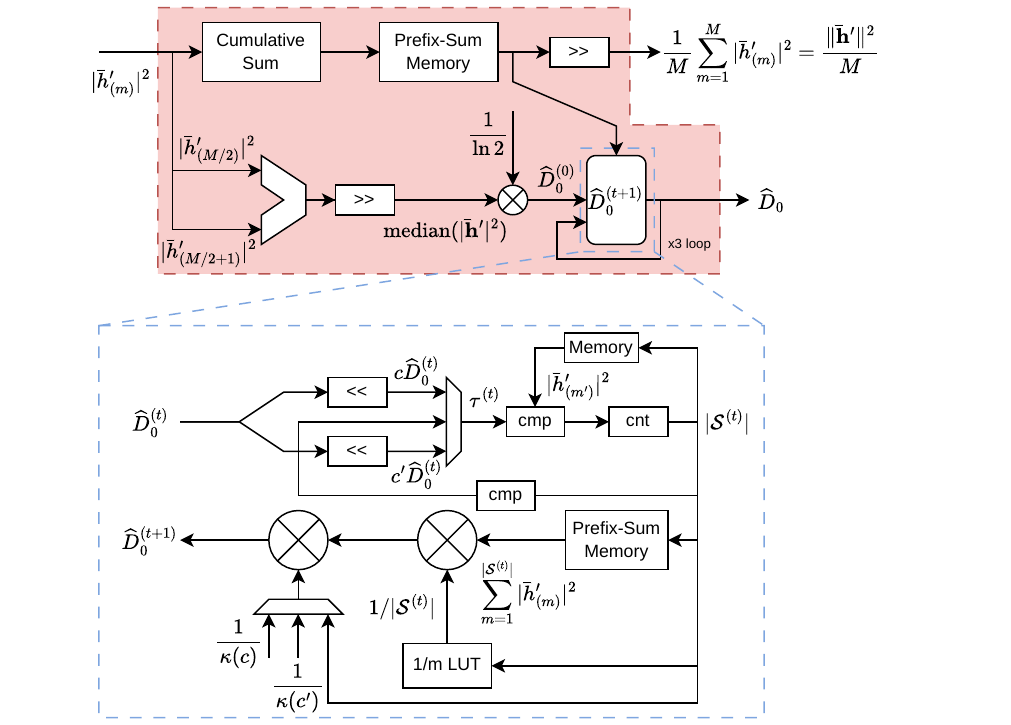}\label{subfig:truncatedmean}}
    \caption{Architecture of the composite noise power estimator. (a) sorting unit (b) truncated mean unit.}
    \label{fig:vlsi_d0}
\end{figure}

Subsequently, as shown in Fig.~\ref{subfig:truncatedmean}, the sorted samples are then streamed into the truncated mean unit, where a cumulative sum is computed and stored as prefix-sums in memory. Let $|\bar{h}^\prime_{(m)}|^2$ denote the $m$-th element in ascending order. The corresponding prefix-sum is given by $\sum_{m^\prime=1}^{m} |\bar{h}^\prime_{(m^\prime)}|^2$, which is subsequently utilized for efficient truncated mean computation.
After the prefix-sums are fully stored in memory, the median is computed using the $(M/2)$-th and $(M/2+1)$-th elements of the sorted sequence. The MAD estimate, which is for the initial value for composite noise variance, is then obtained by multiplying the median by a precomputed constant corresponding to $1/\ln 2$, which is retrieved from a lookup table (LUT). This value serves as the input to the iterative procedure for estimating $\widehat{D}_0$, and this initialization corresponds to line 2 in Algorithm~\ref {alg:noisepwrest}.
The input is subsequently scaled by the parameters $c$ and $c^\prime$, both of which are chosen as powers of two and thus efficiently implemented via bit-shifting operations. The resulting threshold is compared with $|\bar{h}^\prime_{(m)}|^2$. 
If the threshold is larger, the comparison proceeds to the next higher index; otherwise, it moves to the previous index. This approach eliminates the need to scan the entire sequence at each iteration for obtaining the set of noise in \eqref{eqn:setofnoise}. 
When the index falls below $\rho_\mathrm{min}$, the value associated with $c^\prime$ is used instead.
After determining $|\mathcal{S}^{(t)}|$, the corresponding prefix-sum is retrieved from memory using the identified index. 
Using prefix-sum and sorted samples, the computational complexity of lines 7 to 9 in Algorithm~\ref{alg:noisepwrest} can be reduced by avoiding the need to scan all elements.
The estimate is then refined through mean normalization and $\kappa(\cdot)$ normalization, both implemented using LUT-based operations. The resulting value is output as the updated estimate. This iterative process is performed three times, yielding the final estimate $\widehat{D}_0$ for output.

\begin{figure}
    \centering
    \subfloat[]{\includegraphics[width=\linewidth]{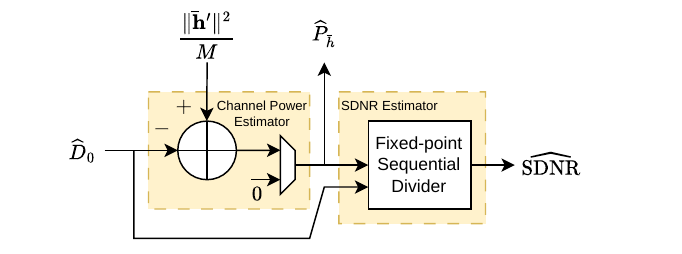}\label{subfig:vlsi_sigsnr}
    }

    \subfloat[]{\includegraphics[width=\linewidth, trim= 40 0 80 0, clip]{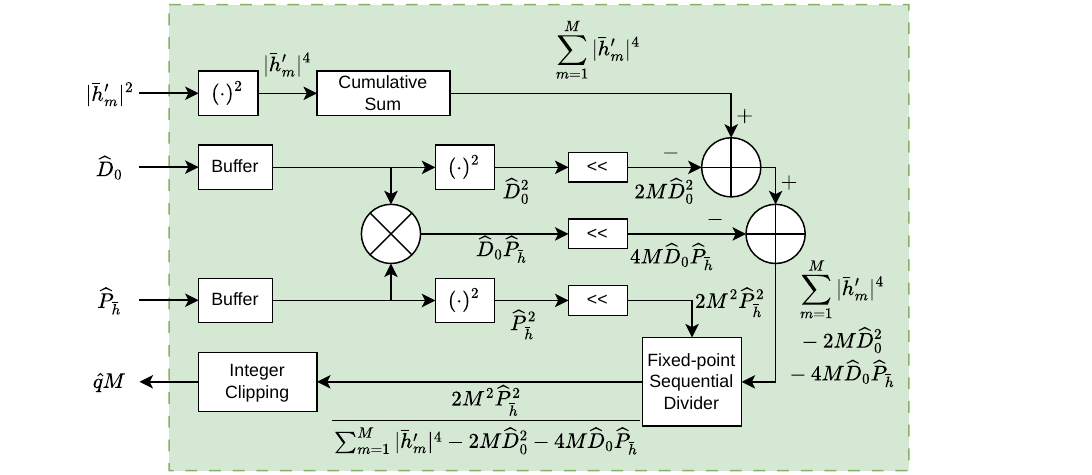}\label{subfig:vlsi_activity}
    }
    \caption{Architecture of the auxiliary estimators. (a) estimator units for average channel power and SDNR (b) activity rate estimator unit.}
    \label{fig:vlsi_auxiliaryest}
\end{figure}

Fig.~\ref{fig:vlsi_auxiliaryest} shows the estimators for average channel power, SDNR, and activity rate. As shown in Fig.~\ref{subfig:vlsi_sigsnr}, the average channel power estimation in \eqref{eqn:channelpwr} is performed, and SDNR, which is the ratio between $\widehat{P}_{\bar{h}}$ and $\widehat{D}_0$, is calculated with a fixed-point sequential pipelined divider. Fig.~\ref{subfig:vlsi_activity} shows the detailed architecture of the activity rate estimator.
The estimate computed in this unit is a modified version of the operation in \eqref{eqn:activityest}. 
The activity rate is inherently bounded, and when multiplied by $M$, it represents the number of active beams, which is an integer-valued quantity.
Accordingly, the output is reformulated to directly produce an integer estimate.
Moreover, directly squaring $\widehat{\mathrm{SDNR}}$ after a fixed-point division may lead to numerical instability, while also introducing unnecessary division operations.
To address this, the expression is reformulated to avoid such operations using  $\mathrm{SDNR}=P_{\bar{h}}/D_0$. As a result, the number of active beams can be estimated as
% \begin{equation}
% \begin{aligned}
\begin{align}
    \hat{q}M & = \argmin_{(\hat{q}M)\in\mathbb{N}} \left| \dfrac{2M(\widehat{\mathrm{SDNR}})^2}{\frac{1}{M\widehat{D}_0^2}\sum_{m=1}^M|\bar{h}_m^\prime|^4  -2-4\widehat{\mathrm{SDNR}}} - \hat{q}M\right| \nonumber\\
    & = \mathrm{round}
     \left(\dfrac{2M^2\widehat{D}_0^2(\widehat{\mathrm{SDNR}})^2}{\sum_{m=1}^M|\bar{h}_m^\prime|^4  -2M\widehat{D}_0^2-4M\widehat{\mathrm{SDNR}}\widehat{D}_0^2} \right)\nonumber\\
    & = \mathrm{round}
     \left( \dfrac{2M^2\widehat{P}_{\bar{h}}^2}{\sum_{m=1}^M |\bar{h}_m^\prime|^4 - 2M\widehat{D}_0^2 - 4M\widehat{D}_0\widehat{P}_{\bar{h}}}
    \right),
\end{align}
% \end{aligned}
% \end{equation}
where $\mathrm{round}(\cdot)$ is the integer clipping function.
To compute the fourth-order moment, the element-wise squared magnitudes obtained in the previous stage are squared once more and accumulated to form the cumulative sum of the fourth-power terms.
During this process, $\widehat{D}_0$ and $\widehat{P}_{\bar{h}}$ are temporarily stored in buffers to ensure proper synchronization with the subsequent arithmetic operations.
Subsequently, multiplications involving powers of two, such as $2$ and $M$, are efficiently implemented using bit-shifting operations, while division operations are carried out using a fixed-point sequential divider. Finally, the result is clipped to the nearest integer and provided as the output.

Fig.~\ref{fig:vlsi_thresholdcalc} shows the detailed architecture of the threshold calculation unit. Since the activity rate estimator outputs $\hat{q}M$ instead of $\hat{q}$, the threshold in \eqref{eqn:realthreshold} can be rewritten in a form that is more convenient for hardware implementation as
% \begin{equation}
    % \begin{aligned}
\begin{align}
        \eta & = \widehat{D}_0 \left(1+\frac{\hat{q}}{\widehat{\mathrm{SDNR}}}\right) \ln \left( \left(1 + \frac{\widehat{\mathrm{SDNR}}}{\hat{q}}\right) \frac{1-\hat{q}}{\hat{q}}C  \right) \nonumber \\
        & = \widehat{D}_0 \left(1+\frac{\hat{q}M}{M\widehat{\mathrm{SDNR}}}\right)
        \ln \left(\left( 1 + \frac{M\widehat{\mathrm{SDNR}}}{\hat{q}M} \right) \frac{M-\hat{q}M}{\hat{q}M}C \right) \nonumber \\
        & = \widehat{D}_0 \left(1+\frac{\hat{q}M}{M\widehat{\mathrm{SDNR}}}\right)
        \Bigg( \ln \left(  1 + \frac{M\widehat{\mathrm{SDNR}}}{\hat{q}M} \right)  \nonumber
        \\ & \quad + \ln(M-\hat{q}M)  - \ln \hat{q}M + \ln C
        \Bigg).
\end{align}
%     \end{aligned}
% \end{equation}
In this unit, only the division by $M\widehat{\mathrm{SDNR}}$ is implemented using a fixed-point sequential divider, while the remaining division operations are replaced by reciprocal LUTs or by exploiting logarithmic properties. Specifically, $\ln(M-\hat{q}M)$ and $\ln(\hat{q}M)$ are obtained via LUTs indexed by the integer input $\hat{q}M$, whereas the other logarithmic operations are performed using piecewise linear approximation. The required coefficients for these approximations are also stored in LUTs, and the final threshold value $\eta$ is then produced.

\begin{figure*}
    \centering
    \includegraphics[width=\linewidth, trim = 45 0 10 0, clip]{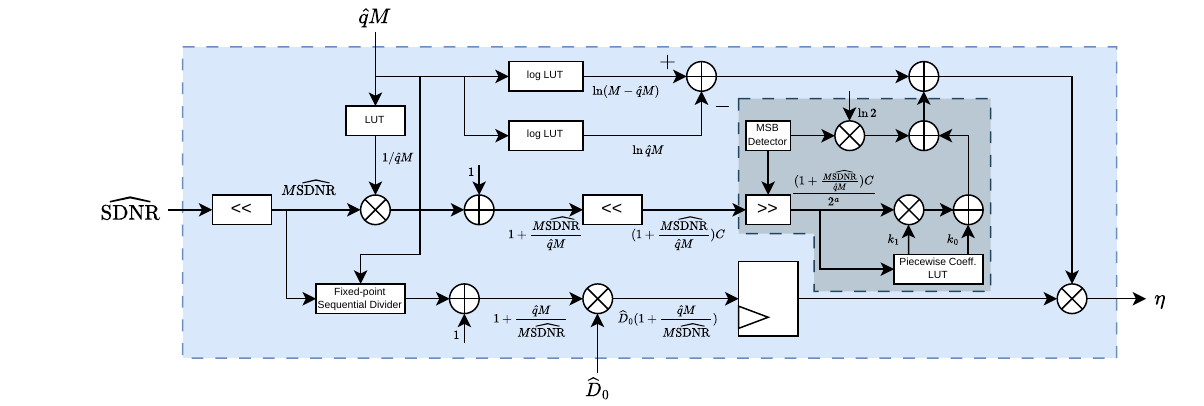}
    \caption{Architecture details of the threshold calculation unit; the gray area represents the piecewise linear approximation of the logarithm.}
    \label{fig:vlsi_thresholdcalc}
\end{figure*}

Fig.~\ref{fig:vlsi_denoise} depicts the architecture of the denoising unit. 
First, the beamspace noisy channel observations obtained from the FFT output, along with their element-wise squared values, are stored in a buffer and held until the computed threshold value $\eta$ becomes available. 
Once the threshold $\eta$ is received, comparison operations are performed to apply hard thresholding. The resulting values are then scaled by the regularization factor $1/\alpha$, which accounts for the quantization gain and is loaded from LUT. 

\begin{figure}
    \centering
    \includegraphics[width=\linewidth]{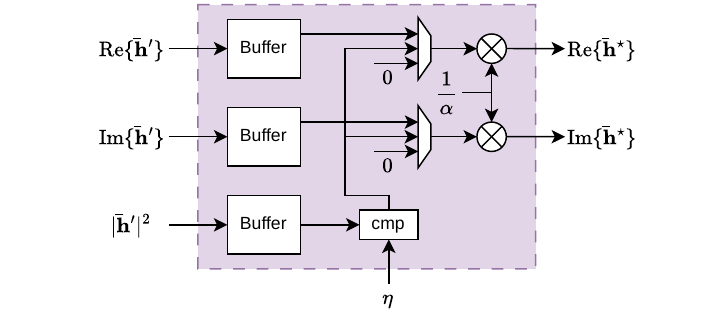}
    \caption{Architecture details of denoising unit.}
    \label{fig:vlsi_denoise}
\end{figure}

\subsection{FPGA Implementation Results}
To verify the practical performance of the proposed method, the VLSI design is implemented using the AMD-Xilinx Kintex UltraScale+ FPGA KCU116 Evaluation Kit.
The implementation is performed across different numbers of BS antennas, and the results are presented in Table~\ref{tbl:implementation1}. 
Overall, the design utilizes only a small fraction of the available FPGA hardware resources, even in the case of $M=128$, which incurs the highest resource consumption.
Moreover, the increase in resource utilization with respect to $M$ is sublinear.
The required number of clock cycles also increases with $M$, primarily due to the presence of modules that operate sequentially across vector elements.
Nevertheless, even for the larger antenna configuration ($M=128$), the latency remains as low as $4.210~\mu$s. Moreover, even when estimating the channels of 16 users, the total processing time is only $67.36~\mu$s, which is significantly shorter than the millisecond-scale coherence time of mmWave channels~\cite{ref:mmwavecoherencetimeshort}.
% In addition, the power consumption is low, measured at 287 mW and 426 mW for 64- and 128-element ULAs, respectively.

The resource utilization breakdown for each module is summarized in Table~\ref{tbl:implementationbreakdown}. The largest portion of the latency is attributed to the preprocessing and postprocessing stages, where the FFT and IFFT operations consume the majority of clock cycles.
In addition, the activity rate estimation unit requires 6 clock cycles, the threshold computation unit for $\eta$ requires 3 clock cycles, and the denoising unit consumes 67 and 131 clock cycles for $M=64$ and $M=128$, respectively.
If prior knowledge of the SNR is available, the estimation processes for $\widehat{D}_0$, $\widehat{P}_{\bar{h}}$, and $\widehat{\mathrm{SDNR}}$ can be omitted. In this case, the overall number of required clock cycles can be reduced by approximately 35\%, while achieving even greater reductions in LUT and flip-flop utilization. 

Table~\ref{tbl:vlsicomparison} compares the proposed hardware implementation with prior art hardware implementations for massive MIMO channel estimation. 
The proposed method is the only one capable of effectively handling the impact of quantization noise arising from low-resolution ADCs. Furthermore, under the same assumption regarding the availability of SNR prior knowledge, it demonstrates significantly lower latency, both in terms of clock cycles and absolute time, compared to existing approaches.

\begin{table}[]
    \caption{FPGA hardware resource utilization results with Xilinx Kintex UltraScale+ FPGA KCU116 Evaluation Kit}
    \label{tbl:implementation1}
    {\centering
    \begin{tabular}{|c|c|c|}
    \hline
    $M$                             & 64    & 128 \\ \hline\hline
    LUTs                            & 7705 (3.55\%)  & 11664 (5.38\%)   \\ 
    LUTs as logic                   & 7327 (3.38\%) & 11138  (5.13\%)  \\ 
    LUTs as memory                  & 378  (0.38\%) & 526 (0.53\%)   \\ 
    Flipflops                       & 10633 (2.45\%) & 17458  (4.02\%)  \\ 
    DSP48 units                     & 27 (1.48\%) & 33  (1.81\%)  \\ 
    Maximum clock frequency (MHz)   & 339   &  314   \\ 
    Latency (clock cycles)          & 770   &  1322 \\ 
    Latency ($\mu$s)                & 2.271 &  4.210   \\
    Dynamic power consumption$^a$ (W)           & 0.287 &  0.426   \\ 
    Maximum throughput$^b$ (Mvectors/s) & 5.297 &  2.453  \\ \hline
    \end{tabular}\par}
    \vspace{2mm}
    {\footnotesize
    $^{a}$The power consumption is obtained from the statistical power estimation at the maximum clock frequency under a supply voltage of 0.85 V. \\
    $^{b}$The maximum throughput is obtained in million vectors processed per second and is calculated as $f_\mathrm{max}/M$, where $f_\mathrm{max}$ denotes the maximum clock frequency.}
\end{table}

\begin{table*}[]
    \caption{FPGA hardware resource utilization breakdown with Xilinx Kintex UltraScale+ FPGA KCU116 Evaluation Kit}
    \label{tbl:implementationbreakdown}
    \centering
    
\begin{tabular}{|c|c|ccccc|}
\hline
                         & Modules                  & Pre-/post-processing & \begin{tabular}[c]{@{}c@{}}Estimators for $\widehat{D}_0$,\\ $\widehat{P}_{\bar{h}}$, and $\widehat{\mathrm{SDNR}}$\end{tabular} & $\hat{q}$ estimator & $\eta$ calculation & Denoising \\ \hline\hline
\multirow{4}{*}{$M=64$}  & LUTs                  & 1858                 & 4340                                                                                                                                & 447                 & 475   &  615                       \\  
                         & Flipflops             & 3189                 & 4940                                                                                                                                & 580                 & 493 & 1397                            \\  
                         & DSP48 units           & 12                   & 3                                                                                                                                   & 4                   & 8  & 0                             \\  
                         & Latency (clock cycle) & 424                  & 270                                                                                                                                 & 6                   & 3 & 67                               \\ \hline\hline
\multirow{4}{*}{$M=128$} & LUTs                  & 2313                 & 7385                                                                                                                                & 463                 & 475 & 1051                            \\  
                         & Flipflops             & 3958                 & 9723                                                                                                                                & 587                 & 495 & 2686                            \\  
                         & DSP48 units           & 18                   & 3                                                                                                                                   & 4                   & 8  & 0                             \\  
                         & Latency (clock cycle) & 700                  & 482                                                                                                                                 & 6                   & 3 & 131                              \\ \hline
\end{tabular}
\end{table*}

\begin{table*}[]
    \caption{FPGA implementation result comparison}
    \label{tbl:vlsicomparison}
    \centering
\begin{tabular}{|c|cc|c|c|c|c|}
\hline
                               & \multicolumn{2}{c|}{This work}             & Mirfarshbafan~\cite{ref:beaches} &  Chundi~\cite{ref:hw4}      &  Liu~\cite{ref:liu}     &   Guo~\cite{ref:hw7}     \\ \hline\hline
Consider low-resolution ADC?    & \multicolumn{2}{c|}{Yes}            &  No                                    &  No       &  No      &  No       \\ 
Require prior knowledge of SNR? & \multicolumn{1}{c}{ No }     & Yes    & Yes                                   &  No       &  No      &  No       \\ 
$M$                            & \multicolumn{2}{c|}{128}          & 128                                 & 64     & 128   & -      \\ 
Clock frequency                & \multicolumn{2}{c|}{314}          & 303                                 & 125    & 217   & 62.5   \\ 
LUTs                           & \multicolumn{1}{c}{11664} & 7324 & 6391                                & 33829  & 24130 & 862261 \\ 
Flipflops                      & \multicolumn{1}{c}{17458} & 7735 & 7015                                & 34735  & 38464 & -      \\ 
DSP units                      & \multicolumn{1}{c}{27}    & 24   & 32                                  & 58     & 432   & 31     \\ 
Latency (clock cycle)          & \multicolumn{1}{c}{1322}  & 840  & 972                                 & 196250 & -     & -      \\ 
Latency ($\mu$s)               & \multicolumn{1}{c}{4.2}   & 2.7  & 3.2                                 & 1570   & -     & -      \\ \hline
\end{tabular}
\end{table*}
\section{Conclusion}\label{sec:conclusion}
In this paper, we proposed a low-complexity beamspace channel denoising algorithm for mmWave massive MIMO systems with low-resolution ADCs. By leveraging the inherent sparsity of mmWave channels in the beamspace domain, the denoising problem was formulated as a Bayesian binary hypothesis testing under a Bernoulli-complex Gaussian framework. 
In particular, thermal noise and quantization noise were jointly modeled as a composite noise, enabling a unified statistical characterization of distortion induced by low-resolution ADCs. Based on this formulation, a closed-form decision rule was derived to distinguish signal-dominant and noise-dominant components, followed by a hard-thresholding-based denoising procedure.
The proposed algorithm avoids computationally intensive operations such as matrix inversion, iterative optimization, and parameter searching, and achieves near-linear computational complexity with respect to the number of antennas. This property makes the proposed method particularly suitable for large-scale mmWave systems where both latency and computational efficiency are critical. 
Furthermore, a hardware-efficient VLSI architecture was developed and implemented on an FPGA platform, demonstrating significantly lower latency and reduced hardware resource utilization compared to existing approaches.
Simulation and implementation results confirmed that the proposed method achieves competitive performance compared to existing approaches while significantly reducing computational complexity and hardware cost compared to conventional hardware implementations. These results highlight the effectiveness of the proposed framework as a practical solution for beamspace channel denoising in mmWave massive MIMO systems with low-resolution ADCs.

\appendices
\section{Derivation of Hypothesis Testing Threshold}\label{app:testderivation}
In this appendix, we derive the Bayesian hypothesis testing threshold in \eqref{eqn:testing_orig}.
First, the left-hand side of \eqref{eqn:testing_orig} can be organized as
% \begin{equation}
% \begin{aligned}
\begin{align}
    & \dfrac{\dfrac{1}{\pi(\alpha^2P_{\bar{h}}/p(\mathcal{H}_1)+D_0)} \exp\left(-\dfrac{|\bar{h}^\prime_m|^2}{\alpha^2P_{\bar{h}}/p(\mathcal{H}_1)+D_0}\right)}
    {\dfrac{1}{\pi D_0}\exp\left(-\dfrac{|\bar{h}_m^\prime|^2}{D_0}\right)} \nonumber \\
    & = \frac{D_0}{\alpha^2P_{\bar{h}}/p(\mathcal{H}_1)+D_0} \nonumber \\
    & \quad \times \exp \left(
    -\dfrac{|\bar{h}^\prime_m|^2}{\alpha^2P_{\bar{h}}/p(\mathcal{H}_1)+D_0
    }
    + \frac{|\bar{h}_m^\prime|^2}{D_0}
    \right).
\end{align}
% \end{aligned}
% \end{equation}
By substituting the left-hand side of \eqref{eqn:testing_orig}, it can be rewritten as
\begin{equation}
    \begin{aligned}
    & \frac{D_0}{\alpha^2P_{\bar{h}}/p(\mathcal{H}_1)+D_0} \\
    & \quad\quad \times \exp \left(
    -\dfrac{|\bar{h}^\prime_m|^2}{\alpha^2P_{\bar{h}}/p(\mathcal{H}_1)+D_0
    }
    + \frac{|\bar{h}_m^\prime|^2}{D_0}
    \right) \\
    & \quad\quad \mathop{\gtrless}_{\mathcal{H}_0}^{\mathcal{H}_1}
    \frac{p(\mathcal{H}_0)}{p(\mathcal{H}_1)}C = \frac{1-p(\mathcal{H}_1)}{p(\mathcal{H}_1)}C.
    \end{aligned}
\end{equation}
Thus, we obtain
\begin{equation}
    \begin{aligned}
    & \exp \left(
    -\dfrac{|\bar{h}^\prime_m|^2}{\alpha^2P_{\bar{h}}/p(\mathcal{H}_1)+D_0
    }
    + \frac{|\bar{h}_m^\prime|^2}{D_0}
    \right) \\
    & \quad\quad  \mathop{\gtrless}_{\mathcal{H}_0}^{\mathcal{H}_1}
    \left(\frac{\alpha^2P_{\bar{h}}/p(\mathcal{H}_1)+D_0}{D_0} \right)\frac{1-p(\mathcal{H}_1)}{p(\mathcal{H}_1)}C.
    \end{aligned}
\end{equation}
By applying the logarithm to both sides, it yields
% \begin{equation}
%     \begin{aligned}
\begin{align}
        & -\dfrac{|\bar{h}^\prime_m|^2}{\alpha^2P_{\bar{h}}/p(\mathcal{H}_1)+D_0} + \frac{|\bar{h}_m^\prime|^2}{D_0} \nonumber \\
        & = |\bar{h}^\prime_m|^2 \left( -\dfrac{1}{\alpha^2P_{\bar{h}}/p(\mathcal{H}_1)+D_0} + \frac{1}{D_0} \right) \nonumber \\
        & = |\bar{h}^\prime_m|^2 \left( \dfrac{-D_0+\alpha^2P_{\bar{h}}/p(\mathcal{H}_1)+D_0}{(\alpha^2P_{\bar{h}}/p(\mathcal{H}_1)+D_0)D_0} \right) \nonumber \\
        & = |\bar{h}^\prime_m|^2 \left( \dfrac{\alpha^2P_{\bar{h}}/p(\mathcal{H}_1)}{(\alpha^2P_{\bar{h}}/p(\mathcal{H}_1)+D_0)D_0} \right) \nonumber \\
        & \mathop{\gtrless}_{\mathcal{H}_0}^{\mathcal{H}_1}
        \ln
        \left(\frac{\alpha^2P_{\bar{h}}/p(\mathcal{H}_1)+D_0}{D_0}\cdot \frac{1-p(\mathcal{H}_1)}{p(\mathcal{H}_1)}C \right). 
\end{align}
%     \end{aligned}
% \end{equation}
Subsequently, we can rewrite the testing equation in terms of $|\bar{h}_m^\prime|^2$ as
\begin{equation}
\begin{aligned}
    |\bar{h}_m^\prime|^2 & \mathop{\gtrless}_{\mathcal{H}_0}^{\mathcal{H}_1} 
    \left( \dfrac{(\alpha^2P_{\bar{h}}/p(\mathcal{H}_1)+D_0)D_0}{\alpha^2P_{\bar{h}}/p(\mathcal{H}_1)} \right) \\
    &\quad \times \ln \left(\frac{\alpha^2P_{\bar{h}}/p(\mathcal{H}_1)+D_0}{D_0}\cdot \frac{1-p(\mathcal{H}_1)}{p(\mathcal{H}_1)}C \right) \\
    & = D_0 \left( 1+\dfrac{D_0}{\alpha^2P_{\bar{h}}/p(\mathcal{H}_1)} \right) \\
    &\quad \times \ln \left(\left(1+\frac{\alpha^2P_{\bar{h}}/p(\mathcal{H}_1)}{D_0}\right) \frac{1-p(\mathcal{H}_1)}{p(\mathcal{H}_1)}C \right) \\
    & = D_0 \left( 1 + \frac{p(\mathcal{H}_1)}{\mathrm{SDNR}} \right) \ln \left(
    \left( 1+\frac{\mathrm{SDNR}}{p(\mathcal{H}_1)}\right)\frac{1-p(\mathcal{H}_1)}{p(\mathcal{H}_1)}C 
    \right),
\end{aligned}
\end{equation}
since $\mathrm{SDNR} = \alpha^2P_{\bar{h}}/D_0$.

\section{Proof of Lemma~\ref{lem:truncatedmeanbias}}\label{app:prooflem}
By the definition of conditional expectation,
\begin{equation}
\mathbb{E}[X \mid X \le k\lambda]
=
\frac{\int_0^{k\lambda} x f_X(x)\,dx}{\Pr(X \le k\lambda)}.
\end{equation}
First, the denominator is given by
\begin{equation}
\Pr(X \le k\lambda)
=
\int_0^{k\lambda} \frac{1}{\lambda} e^{-x/\lambda} dx = 1 - e^{-k}.
\end{equation}
Next, the numerator is
\begin{equation}
    \int_0^{k\lambda} x f_X(x)\,dx
    =
    \int_0^{k\lambda} x \frac{1}{\lambda} e^{-x/\lambda} dx 
    = \lambda \int_0^k u e^{-u} du.
\end{equation}
Using integration by parts, we have
\begin{equation}
\int_0^k u e^{-u} du
=
\left[-(u+1)e^{-u}\right]_0^k
=
1 - (1+k)e^{-k}.
\end{equation}
Therefore, the numerator becomes $\lambda \bigl(1 - (1+k)e^{-k}\bigr)$.
Combining the numerator and denominator, we obtain
\begin{equation}
\mathbb{E}[X \mid X \le k\lambda]
=
\frac{\lambda \bigl(1 - (1+k)e^{-k}\bigr)}{1 - e^{-k}}
=
\lambda \kappa(k) = \mathbb{E}[X]\kappa(k),
\end{equation}
which completes the proof.
\qed

\appendices
% \section{Proof of the First Zonklar Equation}
% Appendix one text goes here.

% % you can choose not to have a title for an appendix
% % if you want by leaving the argument blank
% \section{}
% Appendix two text goes here.

% % use section* for acknowledgment
% \section*{Acknowledgment}

% The authors would like to thank...

% Can use something like this to put references on a page
% by themselves when using endfloat and the captionsoff option.
\ifCLASSOPTIONcaptionsoff
  \newpage
\fi

% trigger a \newpage just before the given reference
% number - used to balance the columns on the last page
% adjust value as needed - may need to be readjusted if
% the document is modified later
%\IEEEtriggeratref{8}
% The "triggered" command can be changed if desired:
%\IEEEtriggercmd{\enlargethispage{-5in}}

% references section

% can use a bibliography generated by BibTeX as a .bbl file
% BibTeX documentation can be easily obtained at:
% http://mirror.ctan.org/biblio/bibtex/contrib/doc/
% The IEEEtran BibTeX style support page is at:
% http://www.michaelshell.org/tex/ieeetran/bibtex/
\bibliographystyle{IEEEtran}
% argument is your BibTeX string definitions and bibliography database(s)
\bibliography{bibtex/bib/IEEEabrv, bibtex/bib/paper}
\end{document}